\documentclass[11pt]{article}
\usepackage{amsmath, amsthm, amssymb, mathtools}
\usepackage{thmtools,thm-restate}
\usepackage{geometry}
\usepackage{hyperref}
\usepackage{cleveref}
\usepackage{natbib}

\newcommand{\Prob}{\mathbb{P}}

\newtheorem{thm}{Theorem}[section]

\newtheorem{prop}{Proposition}[section]
\newtheorem{asmp}{Assumption}[section]
\newtheorem{defn}{Definition}[section]

\newtheorem{rem}{Remark}[section]
\newtheorem{example}{Example}[section]

\def\approxcorrect{\checkmark\kern-1.1ex\raisebox{.89ex}{$\times$}}

\usepackage{amsmath,amsfonts,bm}

\def\eqref#1{equation~\ref{#1}}

\def\1{\bm{1}}

\DeclareMathAlphabet{\mathsfit}{\encodingdefault}{\sfdefault}{m}{sl}
\SetMathAlphabet{\mathsfit}{bold}{\encodingdefault}{\sfdefault}{bx}{n}

\def\gD{{\mathcal{D}}}
\def\gE{{\mathcal{E}}}

\def\gG{{\mathcal{G}}}

\def\gL{{\mathcal{L}}}

\def\gS{{\mathcal{S}}}

\def\gV{{\mathcal{V}}}

\def\gX{{\mathcal{X}}}
\def\gY{{\mathcal{Y}}}
\def\gZ{{\mathcal{Z}}}

\newcommand{\E}{\mathbb{E}}

\newcommand{\R}{\mathbb{R}}

\newcommand{\KL}{D_{\mathrm{KL}}}

\newcommand{\Cov}{\mathrm{Cov}}

\providecommand{\gX}{\mathcal{X}}
\providecommand{\gY}{\mathcal{Y}}
\providecommand{\gL}{\mathcal{L}}
\providecommand{\gG}{\mathcal{G}}
\providecommand{\gE}{\mathcal{E}}
\providecommand{\gZ}{\mathcal{Z}}
\providecommand{\gV}{\mathcal{V}}
\providecommand{\R}{\mathbb{R}}
\providecommand{\E}{\mathbb{E}}
\providecommand{\Prob}{\mathbb{P}}
\providecommand{\KL}{\mathrm{KL}}

\title{The Representation-Rationalizability Tradeoff in Reward Learning}
\date{}
\author{
Jing Dong\\
Vector Institute\\
\texttt{jing.dong@vectorinstitute.ai}
\and
Yaoliang Yu\\
University of Waterloo \& Vector Institute\\
\texttt{yaoliang.yu@uwaterloo.ca}
\and
Pascal Poupart\\
University of Waterloo \& Vector Institute\\
\texttt{ppoupart@uwaterloo.ca}
}
\begin{document}
\maketitle

\begin{abstract}
In RLHF, each training example contains a prompt $x$ and two candidate responses $y,y'$, and annotators provide pairwise preferences between these responses. The learning problem is to convert these heterogeneous pairwise judgments into a single scalar reward $r(x,y)$ that measures response quality for each prompt. Classical social choice implies an impossibility because heterogeneous annotator samples can induce pooled preferences with Condorcet cycles, so no scalar reward can evaluate all compared response pairs consistently. A growing literature analyzes RLHF as a social-choice problem, but usually assumes a fixed finite set of alternatives, i.e., a pre-enumerated finite set of candidate responses for each prompt. Modern pipelines instead score responses through a learned representation $\phi(x,y)$ before a scalar head, so $\phi$ determines which responses are treated as distinguishable alternatives and which comparisons are visible to the reward model. Once this embedding is part of the problem, the impossibility results from social choice theory become a tradeoff. We show that the excess cross-entropy loss of any reward  built on $\phi$ decomposes exactly into a representational term, which a richer $\phi$ shrinks, and an aggregation term, which a richer $\phi$ enlarges by exposing more comparisons that no scalar can rank consistently. The same results extend to direct preference optimization (DPO), and jointly training the embedding and the reward cannot guarantee to recover the sweet spot of this tradeoff. Experiments on synthetic data and real preference datasets corroborate our results.
\end{abstract}

\section{Introduction}

Should reward models use richer representations to better fit human preferences? In modern RLHF  \citep{Ouyangetal22}, the aggregation task is to fit a single scalar reward to pooled pairwise preferences from many annotators. A popular architecture uses a pretrained language-model encoder that maps a prompt-response pair $(x,y)$ to features $\phi(x,y)$, followed by a thin score head $g$, so $r(x,y)=g(\phi(x,y))$. A common intuition is that increasing representational richness can only help because it gives the reward model more signals to fit. We show that this intuition is incomplete. Richer representations can improve fit in one sense while hurting it in another.

The reason is that representation changes the aggregation problem itself. The embedding determines which response pairs are distinguishable to the reward fitter, so refining $\phi$ turns previously collapsed responses into distinct alternatives and adds new pairwise comparison that a single scalar reward must rationalize. Classical social choice theory already warns that no scalar reward can consistently aggregate pooled judgments from heterogeneous annotators when their combined comparisons contain cycles
\cite{arrow2012social,gibbard1973manipulation,satterthwaite1975strategy}, and a recent body of work has translated these social choice results into the RLHF setting \cite{conitzerposition24,dai2024mapping,xiao2025theoretical,liu2025statistical}. Most of these works, however, treat the response space as a fixed finite set of alternatives. Once the embedding is allowed to vary, the social choice impossibility result becomes tunable because changing $\phi$ changes which annotator disagreements are visible in the pooled comparisons.

Our main result makes this trade-off precise. Under a random-utility model with arbitrary annotator heterogeneity, the excess cross-entropy loss of any embedding-based reward decomposes into two parts: an embedding loss measuring information that $\phi$ discards by collapsing distinguishable preferences, and an agreement cost that measures how far the pooled preferences, as seen through $\phi$, are from being rationalizable by any scalar reward. Under a mild assumption on the representation, the lower bounds of the two terms move in opposite directions as $\phi$ refines. As a consequence, both the smallest and the largest representations can be suboptimal, and the optimal embedding dimension is dataset-dependent. We further show that this lower bound is tight under a constructive reward.

Our framework extends beyond reward modeling. We show that Direct Preference Optimization  (DPO) \citep{rafailov2023direct} admits a similar decomposition and excess-loss lower bound, demonstrating that the same tradeoff governs preference-based policy learning even without an explicit reward model. Based on the excess loss lower bound results, one might hope that jointly training the embedding and reward can naturally converge to an optimal embedding. However, we show that the joint embedding-reward training is not guaranteed to decrease either tradeoff term per round and admits only fixed points where the reward is constant on the data.

We complement the theory with experiments on synthetic data and three real preference datasets (Jester, Sushi, and MT-Bench) using two pretrained sentence encoders. Across settings, the embedding loss decreases and the agreement cost increases with embedding resolution, and total excess loss is often minimized at an intermediate dimension, consistent with our bounds.
\section{Related Work}
\paragraph{Alignment with Social Choice}
Social choice theory offers a rich framework for analyzing how RLHF aggregates diverse human preferences. The optimal solution to the Bradley-Terry loss in RLHF induces a ranking equivalent to the Borda count from social choice theory~\citep{siththaranjandistributional24}. Under the assumption that each pairwise comparison is labeled by a single evaluator, RLHF satisfies pairwise majority and Condorcet consistency~\citep{xiao2025theoretical}. More fundamentally, as the annotator pool grows, the probability that the preference profile contains a Condorcet cycle approaches one exponentially fast~\citep{liu2025statistical}, so reward-based alignment is statistically fragile under general preferences. In the case of a linear social choice model, where voters are linear directions in fixed embedding space, \citet{ge2024axioms} analyze the ranking induced by an optimal linear reward with respect to Pareto optimality (PO) and Pairwise Majority Consistency (PMC). \citet{hollenderenforcing} show that an appropriate reweighting of the Bradley-Terry loss can enforce PO and PMC, and further analyze clone robustness. 

\paragraph{RLHF Reward Modeling, and Failure Modes}
Reward models are central to RLHF yet remain difficult to train reliably. \citet{casper2023open} give a survey summarizing limitations of RLHF, including reward hacking, reward misgeneralization, distribution shift, and the difficulty of representing diverse human values with a single reward. \citet{wang2024secrets} identify noisy preference pairs and out-of-distribution generalization as practical challenges. A number of works identify reward overfitting and overoptimization in reward-model training and propose alternative training methods to standard RLHF \citep{rafailov2024scaling,zhu2024iterative,yang2024regularizing,costereward}. \citet{gao2023scaling} quantify how proxy reward scores diverge from true rewards as optimization pressure increases, following smooth scaling laws.

On the theoretical side, \citet{zhu2023principled}  show that with a linear reward model and a Bradley-Terry (BT) model, MLE can converge. A number of works propose alternatives to the Bradley-Terry model as scalar rewards cannot capture intransitive preferences \citep{munos2024nash,azar2024general,zhang2025beyond}. \citet{siththaranjandistributional24} show that implicitly aggregating over hidden contexts during reward modeling can produce counter-intuitive results very different from aggregation via expected utility. \citet{sun2024rethinking} establish the asymptotic convergence rate of BT reward models based on deep neural networks using embeddings.

\section{RLHF Pipeline}
We now introduce the setup used throughout the paper. We consider the standard RLHF pipeline with $n$ human annotators, denoted as $\{1,\dots,n\}$, a prompt space $\gX$, and a response space $\gY$.
\paragraph{Preference Model}
Fix a prompt $x\in\gX$ and an annotator $i\in\{1,\dots,n\}$.
Annotator $i$ has a latent utility $u_i(x,y)\in\R$ for each response $y\in\gY$.
Given two candidate responses $y,y^\prime\in\gY$, the annotator's observed preference is determined by a Gumbel perturbation~\citep{bradley1952rank,plackett1975analysis}
\begin{align}
  \tilde u_i(x,y) = u_i(x,y)+\varepsilon_i, \qquad \varepsilon_i\stackrel{\mathrm{i.i.d.}}{\sim}\mathrm{Gumbel}(0,\tau),
\end{align}
and the annotator prefers $y$ over $y^\prime$ whenever $\tilde u_i(x,y)>\tilde u_i(x,y^\prime)$.
We note that this Gumbel noise setup is the standard random utility model (RUM) for discrete choice~\citep{luce1959individual,mcfadden1973conditional}. Moreover, as the difference of two independent Gumbel$(0,\tau)$ variables is Logistic$(0,\tau)$, we can recover the Bradley-Terry choice probability as the pairwise preference probability. Specifically, the probability that annotator $i$ prefers $y$ over $y^\prime$ given $x$ is then
\begin{align*}
  P_i(y\succ y^\prime \mid x) = F\left(\frac{u_i(x, y)-u_i(x, y^\prime)}{\tau}\right),
  \qquad F(t)=\frac{1}{1+e^{-t}}\,.
\end{align*}
We set $\tau = 1$ for the rest of the paper for simplicity.

To aggregate across annotators we define the pooled win probability $\bar{p}(x, y, y^\prime) = \frac{1}{n}\sum_{i=1}^n P_i(y\succ y^\prime \mid x)$,
which gives the average probability that annotators prefer $y$ to $y^\prime$. 

\paragraph{Reward Learning}
The reward model is a scalar function $r:\gX\times\gY\to\R$ of the form $r(x,y)=g(\phi(x,y))$, where $\phi:\gX\times\gY\to\R^d$ is a latent embedding and $g:\R^d\to\R$ is a score head. Let $\nu$ be a distribution over prompts and, for each $x$, let $\mu_x$ be a (conditional) distribution over responses.
The reward is learned by minimizing the following negative log-likelihood of the pooled preferences:
\begin{align*}
  \gL(r) = \E_{\substack{x\sim\nu \\ (y,y^\prime)\sim\mu_x^{\otimes 2}}}
  \Bigl[
    -\bar p(x,y,y^\prime)\log F(\Delta r(x,y,y^\prime))
    -\bigl(1-\bar p(x,y,y^\prime)\bigr)\log\bigl(1-F(\Delta r(x,y,y^\prime))\bigr)
  \Bigr]\,,
\end{align*}
where $\Delta r(x,y,y^\prime)= r(x,y)-r(x,y^\prime)$.

\section{The Representation-Rationalizability Tradeoff}\label{sec:tradeoff}
With the setup in place, we now turn to the fundamental question: 
\begin{center}
    \textit{What determines how well a reward model can perform?}
\end{center}
This section gives the formal answer. We first show an exact decomposition of excess loss, then derive a lower bound that makes the representation-rationalizability tension explicit.

Before stating the main theorem, it is instructive to examine a simple example.
\begin{example}\label{ex:cycle}
Fix one prompt and $\gY=\{y_1,y_2,y_3\}$ with uniform $\mu$. Let responses have binary features $\psi(y_1)=(1,0),\quad \psi(y_2)=(0,1),\quad \psi(y_3)=(1,1)$, and suppose pooled preferences satisfy $\bar p(y_1,y_2)=\bar p(y_2,y_3)=\bar p(y_3,y_1)=\tfrac{2}{3}$.
The role of $\phi$ is then to decide which responses can receive different rewards, since $r=g\circ\phi$ is constant on each level set of $\phi$. Now consider two kinds of embeddings, one that uses only the first bit and one that uses both bits. If $\phi_a(y)=\psi(y)_1$, then $\phi_a(y_1)=\phi_a(y_3)$, so $r(y_1)=r(y_3)$ and necessarily $ F(\Delta r(y_1,y_3))=\frac{1}{2}\neq \bar p(y_1,y_3)=\frac{1}{3}$. If $\phi_b(y)=\psi(y)$, then $\phi_b$ is injective and the representation is no longer the bottleneck. But any scalar reward still satisfies $ \Delta r(y_1,y_2)+\Delta r(y_2,y_3)+\Delta r(y_3,y_1)=0$, whereas the preferences imply
\[
  F^{-1}(\bar p(y_1,y_2))+F^{-1}(\bar p(y_2,y_3))+F^{-1}(\bar p(y_3,y_1))
  =3\log 2\neq 0.
\]
Hence no scalar reward can rationalize all three comparisons simultaneously. In summary, a coarse $\phi$ fails by collapsing distinguishable responses, while a fine $\phi$ exposes an irreducible cycle that scalar rewards cannot fit.
\end{example}
To better characterize this tradeoff, we turn to the reward excess loss.

The Bayes-optimal loss $\gL^\ast$ represents the best possible performance achievable by any model that perfectly knows the pooled preference probabilities $\bar{p}(x,y,y^\prime)$. For any learned reward $r(x,y)=g(\phi(x,y))$, the excess loss $\gL(r)-\gL^\ast$ measures how far we fall short of this ideal. Our key observation is that this excess loss admits the following decomposition.
\begin{restatable}{prop}{decomposition}\label{prop:decomposition}
Define the embedding-induced win probability as
\begin{align*}
  \tilde{p}_\phi(x,y,y^\prime) = \E_{(y_0,y_0')\sim\mu_x^{\otimes 2}}\bigl[\bar{p}(x,y_0,y_0')\mid\phi(x,y_0)=\phi(x,y),\;\phi(x,y_0')=\phi(x,y^\prime)\bigr]\,.
\end{align*}
Then for any reward $r(x,y)=g(\phi(x,y))$,
\begin{align*}
  \gL(r)-\gL^\ast
  &= \underbrace{\E_{x\sim\nu}\E_{(y, y^\prime)\sim\mu_x^{\otimes 2}}\left[
    \KL\left(\bar{p}(x,y, y^\prime)\big\|\tilde{p}_\phi(x,y, y^\prime)\right)
  \right]}_{\gE_{\mathrm{emb}}(\phi), \text{embedding loss}}\\
  &\quad + \underbrace{\E_{x\sim\nu}\E_{(y, y^\prime)\sim\mu_x^{\otimes 2}}\left[
    \KL\left(\tilde{p}_\phi(x,y, y^\prime)\big\|F(\Delta r(x,y, y^\prime))\right)
  \right]}_{\gE_{\mathrm{agr}}(\phi,r), \text{agreement cost}}\,.
\end{align*}
\end{restatable}
The embedding loss $\gE_{\mathrm{emb}}$ captures the information discarded when the embedding collapses distinguishable responses onto the same point. The agreement cost $\gE_{\mathrm{agr}}$ captures something fundamentally different. Even if the embedding perfectly preserves all distinctions between responses, the pooled preferences may fail to be rationalizable by any scalar reward. The agreement cost measures this irreducibility. When the embedding is injective, i.e., $\tilde{p}_\phi=\bar{p}$, the embedding loss vanishes entirely. The excess loss is then pure agreement cost. Conversely, when the embedding is constant, i.e., $\tilde{p}_\phi=\frac{1}{2}$ everywhere, the embedding loss is large. But in this case the agreement cost vanishes because a constant win probability is consistent with any reward. In between these extremes, both terms are generally positive.

\begin{defn}\label{def:feasible}
Fix an embedding $\phi_d:\gX\times\gY\to\R^d$. For each $(x,y,y^\prime)$, define $\eta_{d}(x,y,y^\prime)=\min\left\{\tilde{p}_{\phi_d}(x,y,y^\prime),1-\tilde{p}_{\phi_d}(x,y,y^\prime)\right\}$
and the feasibility class
\begin{align*}
  \mathcal{R}_d=\left\{(x,y)\mapsto g(\phi_d(x,y)):\  F\bigl(g(\phi_d(x,y))-g(\phi_d(x,y^\prime))\bigr)\in[\eta_d(x,y,y^\prime),1-\eta_d(x,y,y^\prime)] \right\}.
\end{align*}
\end{defn}
The restriction to $r\in\mathcal{R}_d$ asks only that the implied pairwise probabilities $F(\Delta r)$ stay within the interval $[\eta_d,1-\eta_d]$ that already contains $\tilde p_{\phi_d}$, equivalently that $|\Delta r|\le \log\tfrac{1-\eta_d}{\eta_d}$. This is a mild boundedness condition that is automatically satisfied by any reward whose log-odds match those of the embedding-induced probabilities. 

The above tradeoff is captured by the following lower bound, and its proof is deferred to Appendix~\ref{appendix:tradeoff}.
\begin{thm}\label{thm:mainlb}
For any dimension $d\ge 0$, any embedding $\phi_d:\gX\times\gY\to\R^d$, any $r\in\mathcal{R}_d$, and any $\delta\in(0,\frac{1}{2})$,
\begin{align}\label{eq:mainlb}
  \gL(r)-\gL^\ast \geq 2e_d+\tfrac{\ell_\delta^2}{2}\bar P_d^\delta\,,
\end{align}
where $e_d=\E[(\bar p-\tilde p_{\phi_d})^2]$, $\bar P_d^\delta$ is the curvature-weighted Condorcet cycle probability at margin $\delta$ defined in Definition~\ref{def:condorcet}, $\ell_\delta=F^{-1}(\frac{1}{2}+\delta)=\log\frac{1+2\delta}{1-2\delta}>0$.
\end{thm}
In the following sections, we show that under a mild assumption, $e_d$ is monotonically non-increasing in $d$, while a lower bound on $\bar P_d^\delta$ is non-decreasing in $d$. This generalizes and quantifies the intuition built in Example~\ref{ex:cycle}.

\begin{rem}[Upper bound]
The upper-bound result is deferred to Appendix~\ref{appendix:upper}. With the same decomposition, the embedding loss term's upper and lower bounds are matched up to constants. The agreement cost term is upper bounded by a quantity of the same form as $\bar P_d^\delta$.
\end{rem}

\begin{rem}[DPO extension]
The same decomposition and cycle-based lower-bound proof extends to DPO. The full setup, formal statements, and assumptions are deferred to Appendix~\ref{appendix:dpo}.
\end{rem}

\subsection{Monotonicity Under Representation Refinement}\label{subsec:representation}
The embedding loss depends on the embedding's expressiveness. Intuitively, increasing the embedding dimension should improve the embedding's ability to distinguish responses, which ought to reduce the embedding loss. This intuition is indeed correct under a mild condition, namely that if a response pair looks identical under $\phi_d:\gX\times\gY\to\R^{d}$, it also looks identical under $\phi_{d-1}:\gX\times\gY\to\R^{d-1}$.

This condition is satisfied in many settings. If each $\phi_d$ is a linear embedding $\phi_d(x,y)=U_d^\top\psi(x,y)$ where $\psi:\gX\times\gY\to\R^{D}$ is a fixed feature map and the column spaces satisfy $\operatorname{col}(U_d)\subseteq\operatorname{col}(U_{d+1})$, then higher dimensions add new basis functions and the nesting holds automatically. The same holds when $\phi_d$ concatenates the first $d$ layers or coordinates of a learned representation, so that increasing $d$ only reveals additional information.

We formalize this as follows.
\begin{asmp}\label{asmp:filtration}
For each $d\geq0$, let $\phi_d:\gX\times\gY\to\R^{d}$ be measurable and let $\gG_d$ be the sigma-field generated by the map $\big(\phi_d(x,y), \phi_d(x, y')\big)$ on $\gX\times\gY^2$. Assume $\gG_d\subseteq\gG_{d+1}$ for every $d\geq0$.
\end{asmp}

Under this assumption, the embedding loss is non-increasing in dimension.
\begin{restatable}[]{lem}{monoed}\label{lem:monoed}
Under Assumption~\ref{asmp:filtration}, define $$e_d
  = \E_{x\sim\nu}\,\E_{(y, y^\prime)\sim\mu_x^{\otimes 2}}\left[\left(\bar{p}(x,y, y^\prime)-\tilde{p}_{\phi_d}(x,y, y^\prime)\right)^2\right]$$ then $e_{d+1}\leq e_d$ for all $d \geq 0$. 
\end{restatable}
The proof (deferred to Appendix~\ref{appendix:representation}) uses the tower property as conditioning on a coarser sigma-field yields larger mean-squared error. Since $\gG_d\subseteq\gG_{d+1}$, conditioning on $\gG_{d+1}$ and then on $\gG_d$ is the same as conditioning on $\gG_d$ alone, which cannot reduce the variance of $\bar{p}$.

Having established the monotonicity of the error, we can lower bound the first term as follows.

\begin{prop}\label{prop:ed}
For every measurable $\phi_d$ and every $d\geq0$, $\gE_{\mathrm{emb}}(\phi_d)\geq 2\,e_d$.
\end{prop}

\begin{proof}
Pinsker's inequality gives $\KL(p\|q)\geq 2\operatorname{TV}(p,q)^2$ for distributions $p,q$. For Bernoulli parameters $p,q\in(0,1)$, the total variation distance is $\operatorname{TV}(\mathrm{Bern}(p),\mathrm{Bern}(q))=\tfrac{1}{2}\bigl(|p-q|+|(1-p)-(1-q)|\bigr)=|p-q|$, hence $\KL(p\|q)\geq 2(p-q)^2$. Setting $p=\bar{p}(x,y, y^\prime)$ and $q=\tilde{p}_{\phi_d}(x,y, y^\prime)$ and taking expectations on both sides yields the result.
\end{proof}

\subsection{Cycle-Induced Irrecoverability}\label{subsec:cycle}
The embedding loss captures the information lost to coarsening. But even with a perfectly informative embedding, a more subtle failure mode persists. Suppose $\phi_d$ is injective, so that $\tilde{p}_{\phi_d}=\bar{p}$. The agreement cost $\gE_{\mathrm{agr}}(\phi_d,r)$ measures whether the pooled preferences can be simultaneously rationalized by a scalar reward. When they cannot, we say the preferences are irrecoverable.

The main limitation is that scalar rewards impose transitive structure on preferences.  Define
\begin{align}\label{eq:wd}
  W_d(x,y,y') := F^{-1}(\tilde p_{\phi_d}(x,y,y')).
\end{align}
If $\phi_d$ is injective, then $\tilde p_{\phi_d}=\bar p$, and any reward of the form $r(x,y)=g(\phi(x,y))$ that fits the pooled preference must satisfy
\[
\Delta r(x,y,y') = g(\phi(x,y)) - g(\phi(x,y')) = W_d(x,y,y') = F^{-1}(\bar p(x,y,y')).
\]
The right-hand side is a difference of potentials, which implies a cyclic consistency condition. 

This phenomenon has a direct analogue in social choice theory. Condorcet's paradox~\citep{de2014essai} describes the situation where majority voting over three alternatives produces a cycle: A beats B, B beats C, but C beats A. 

To formalize when this cyclic inconsistency occurs, we introduce the following definition.
\begin{defn}
For a  triplet $(y_1,y_2,y_3)\in\gY^3$ and fixed $x\in\gX$, define the cycle sum $$ C_d(x,y_1,y_2,y_3)
  = W_d(x,y_1,y_2)+W_d(x,y_2,y_3)+W_d(x,y_3,y_1)\,.$$
\end{defn}
Note that if $W_d(x,y,y^\prime)=r(x,y)-r(x,y^\prime)$ for some reward $r$, then $C_d(x,y_1,y_2,y_3)=0$ for all $(y_1,y_2,y_3)$.  Conversely, a non-zero cycle sum prevents any scalar reward from rationalizing the preferences.

We then show that this cycle-sum quantity gives a lower bound on the agreement cost.
\begin{restatable}{thm}{cyclebound}\label{thm:cyclebound}
With $\eta_d$ and $\mathcal{R}_d$ as in Definition~\ref{def:feasible}, define $\kappa_d(x,y,y^\prime)=\eta_d(x,y,y^\prime)\bigl(1-\eta_d(x,y,y^\prime)\bigr)$ and $\kappa_d^{\min}(x,y_1,y_2,y_3)=\min\bigl\{\kappa_d(x,y_1,y_2),\,\kappa_d(x,y_2,y_3),\,\kappa_d(x,y_3,y_1)\bigr\}$.
Then, for any $r\in\mathcal{R}_d$,
\begin{align*}
  \gE_{\mathrm{agr}}(\phi_d,r)
  \geq
  \frac{1}{18}\;
  \E_{x\sim\nu}\,\E_{(y_1,y_2,y_3)\sim\mu_x^{\otimes 3}}
  \Bigl[
    \kappa_d^{\min}(x,y_1,y_2,y_3)
    \,C_d(x,y_1,y_2,y_3)^2
  \Bigr].
\end{align*}
\end{restatable}

Theorem~\ref{thm:cyclebound} shows that the agreement cost is lower-bounded by the minimum-weighted mean squared cycle sum. 

To interpret the bound more concretely, we connect the mean squared cycle sum to the familiar notion of Condorcet cycles.

\begin{defn}[$\delta$-Condorcet cycle]\label{def:condorcet}
Fix $\delta\in(0,\tfrac{1}{2})$. A triplet $(y_1,y_2,y_3)$ forms a positive $\delta$-Condorcet cycle under $\phi_d$ at prompt $x$ if $$\tilde{p}_{\phi_d}(x,y_1,y_2),
  \tilde{p}_{\phi_d}(x,y_2,y_3),
  \tilde{p}_{\phi_d}(x,y_3,y_1)\geq \tfrac{1}{2}+\delta\,,$$
and a negative one if each inequality is reversed ($\leq\tfrac{1}{2}-\delta$). Let $\mathcal{CC}_d^\delta$ be the union of these two events. Define the (unweighted) point-level cycle probability and the $\kappa$-weighted population-level cycle probability by
\begin{align*}
    P_d^\delta(x,y_1,y_2)
  = \ & 
  \Prob_{y_3\sim\mu_x}\left(
    \mathcal{CC}_d^\delta(x,y_1,y_2,y_3)
  \right),
  \quad
  P_d^\delta = 
  \E_{x\sim\nu}\E_{(y_1,y_2)\sim\mu_x^{\otimes 2}}\bigl[P_d^\delta(x,y_1,y_2)\bigr],
  \\
  \bar P_d^\delta
  = \ & 
  \E_{x\sim\nu}\,\E_{(y_1,y_2,y_3)\sim\mu_x^{\otimes 3}}
  \Bigl[
    \kappa_d^{\min}(x,y_1,y_2,y_3)\,\mathbf{1}_{\mathcal{CC}_d^\delta}(x,y_1,y_2,y_3)
  \Bigr]\,.
\end{align*}
Recall that  $\ell_\delta=F^{-1}(\tfrac{1}{2}+\delta) =\log\frac{1+2\delta}{1-2\delta}>0$.
\end{defn}

The population-level cycle probability $\bar P_d^\delta$ is the weighted probability that a random triplet forms a Condorcet cycle, where the weight $\kappa_d^{\min}$ 
downweights cycles whose weakest pair is near 50/50 (which is inherently ambiguous). When a triplet lies in $\mathcal{CC}_d^\delta$, each of its three pairs has a margin of at least $\delta$, so the cycle sum $C_d$ has magnitude at least $3\ell_\delta$.

With this, we lower bound the agreement cost term as follows.
\begin{restatable}{lem}{condorcetlb}\label{lem:condorcetlb}
For every $\delta\in(0,\tfrac{1}{2})$, $\inf_{r\in\mathcal{R}_d}\gE_{\mathrm{agr}}(\phi_d,r) \geq \frac{\ell_\delta^2}{2}\bar P_d^\delta$.
\end{restatable}

To show how the cycle probability evolves with embedding dimension, we first make the following definition.
\begin{defn}
Define $P_\infty^\delta$ as $P_d^\delta$ but with $\bar{p}$ in place of $\tilde{p}_{\phi_d}$ in Definition~\ref{def:condorcet} (same margin conditions on the three pairs $(y_1,y_2)$, $(y_2,y_3)$, and $(y_3,y_1)$). Also define $$\bar P_{\infty}^{\inf,\delta}
=
\inf_{d\ge 0}\;
\E_{x\sim\nu}\,\E_{(y_1,y_2,y_3)\sim\mu_x^{\otimes 3}}
\Bigl[
  \kappa_d^{\min}(x,y_1,y_2,y_3)\,\mathbf{1}_{\mathcal{CC}_\infty^\delta}(x,y_1,y_2,y_3)
\Bigr]\,.$$
\end{defn}

\begin{restatable}{lem}{condorcetgrowth}\label{lem:condorcetgrowth}
Assume $\phi_0$ is constant. For every $\delta\in(0,\tfrac{1}{4})$. Then
\begin{enumerate}
    \item $\bar P_0^\delta=0$, 
    \item For all $d\geq 0$, $\bar P_d^\delta\geq \bar P_\infty^{\inf,2\delta} - \frac{3}{\delta^2}e_d$, 
    \item The map $d\mapsto \bar P_\infty^{\inf,2\delta}-\frac{3}{\delta^2}e_d$ is non-decreasing in $d$.
\end{enumerate}

\end{restatable}
The first condition says that a constant embedding sees no cycle, which is trivially true as $\tilde{p}_{\phi_0}=\frac{1}{2}$ on every pair. The second condition gives a direct lower bound on the weighted cycle probability that appears in Lemma~\ref{lem:condorcetlb}. Roughly, if a triplet is a population-level cycle with margin $2\delta$, it remains a cycle under $\phi_d$ unless the embedding error $e_d$ is large. The third item follows from the monotonicity of $e_d$.

\subsection{Quantitative Lower Bound Under H\"older Utilities}\label{subsec:holder}
The generic lower bound from the previous section holds for any utility functions and any embedding $\phi_d$, and it only shows that there is a tradeoff for the excess loss. However, it is not immediately clear whether the minimum of the excess loss is unique. In particular, while the lower bound of $\bar P_d^\delta$ is non-decreasing in $d$, $\bar P_d^\delta$ might not be. 

To derive more interpretable bounds we now contextualize it under two specific assumptions, a H\"{o}lder continuity condition on the annotator utilities and an $\epsilon$-separating property on the embedding.

Fix a metric $d_\gY$ on $\mathcal{Y}$, such as normalized Hamming or edit distance. We consider the following class of continuous utilities.
\begin{asmp}[$\alpha$-H\"{o}lder utilities]\label{asmp:holder}
There exist $\alpha\in(0,1]$ and $L_u\ge0$ such that for all $x\in\gX$, all $i\in\{1,\dots,n\}$, and all $y,\tilde y\in\gY$, $ |u_i(x,y)-u_i(x,\tilde y)|  \leq L_u\,d_\gY(y,\tilde y)^\alpha$. 
\end{asmp}

The parameter $\alpha\in(0,1]$ interpolates between Lipschitz ($\alpha=1$) and uniformly bounded pairwise differences ($\alpha\to0^+$). 

To measure how much $\phi_d$ coarsens the response space we use the following definition.
\begin{defn}[$\varepsilon$-separating]\label{def:eps-sep}
The embedding $\phi_d$ is $\varepsilon$-\emph{separating} if for all $x\in\gX$, $\phi_d(x,y)=\phi_d(x,\tilde y)$, then $d_\gY(y,\tilde y)\leq\varepsilon$.
\end{defn}

When $\varepsilon=0$ this reduces to injectivity. Larger $\varepsilon$ allows the embedding to identify responses that are farther apart in $d_\gY$. The embedding error $e_d$ is then controlled by the separation parameter. $\varepsilon$-separating embeddings of arbitrarily small $\varepsilon$ can be realized by finite-width ReLU MLPs under mild conditions and we defer the proof to Appendix~\ref{appendix:holder}.

\begin{restatable}{lem}{edupper}\label{lem:ed-upper}
Under Assumption~\ref{asmp:holder}, if $\phi_d$ is $\varepsilon$-separating, then $e_d \le \frac{L_u^2}{4}\varepsilon^{2\alpha}$.
\end{restatable}

Combining Lemma~\ref{lem:ed-upper} with the generic lower bound yields the following result. 
\begin{restatable}{thm}{excess}\label{thm:excess}
Under Assumption~\ref{asmp:holder}, suppose $\phi_d$ is $\varepsilon$-separating. Then for any reward $r\in\mathcal{R}_d$ and any $\delta\in(0,\tfrac{1}{4})$, 
\[
  \gL(r)-\gL^\ast
    \geq
  \frac{1}{2}\,\ell_\delta^2\,\bar P_\infty^{\inf,2\delta}
   - 
  \left(\frac{3\,\ell_\delta^2}{2\,\delta^2}-2\right)^{+}
  \frac{L_u^2}{4\tau^2}\,\varepsilon^{2\alpha},
\]
where $(x)^{+}=\max(0,x)$. 
\end{restatable}

H\"older regularity converts the lower bound in~\eqref{eq:mainlb} into an explicit quantitative bound with three implications. First, the term $\frac{1}{2}\ell_\delta^2\bar P_\infty^{\inf,2\delta}$ is determined by the population cycle geometry of $\bar p$ and therefore defines a data-dependent lower floor that representation design alone cannot eliminate. Second, coarsening can reduce cyclic inconsistency only through an approximation error of order $\varepsilon^{2\alpha}$. Larger $\alpha$ (smoother utilities) yields a more favorable tradeoff, while smaller $\alpha$ limits the benefit of coarsening relative to the induced information loss. Third, $\delta$ specifies the cycle margin at which inconsistency is evaluated. Increasing $\delta$ restricts attention to stronger cycles while decreasing $\delta$ includes weaker cycles but increases sensitivity to representation error through the $1/\delta^2$ scaling. The choice of $\delta$ is therefore a calibration parameter tied to the margin distribution of cycles. Consequently, the minimizing representation dimension is intrinsically dataset-dependent.

\section{Joint Embedding-Reward Updates}\label{sec:joint}
In practice, embeddings and reward heads are typically trained simultaneously, and one might hope that joint training implicitly navigates the representation-rationalizability tradeoff and lands at the sweet spot of our lower bound. We show two structural facts that make this hope fragile.

At iteration $t$ we have a response distribution $\mu^{(t)}_x$ (with joint $\gD_t=\E_{x\sim\nu}(\mu_x^{(t)})^{\otimes 2}$), an embedding $\phi_d^{(t)}$ and a reward $r_t=g_t\circ\phi_d^{(t)}$ minimizing $\gL(\cdot,\gD_t)$, and the soft update for the prompt distribution $d\mu_x^{(t+1)}/d\mu_x^{(t)}(y)\propto e^{\beta r_t(x,y)}$ for a temperature $\beta>0$. Write $e_d(\mu,\phi)=\E_{x\sim\nu}\E_{(y,y^\prime)\sim\mu_x^{\otimes 2}}[(\bar p-\tilde p_\phi)^2]$ for the embedding error at $(\mu,\phi)$, $P_\infty^{2\delta}(\mu)=\Prob_{\nu\otimes\mu^{\otimes 3}}(\mathcal{CC}_\infty^{2\delta})$ for the embedding-independent cycle probability under $\mu$, and $\rho_{r,\beta}(x,y)=e^{\beta r(x,y)}/\E_{\mu_x}[e^{\beta r(x,\cdot)}]$ for the soft-update density. Define the tilt covariances
\begin{align*}
\Gamma(\mu,\phi;r,\beta) &= \E_\nu\Bigl[\Cov_{\mu_x^{\otimes 2}}\bigl((\bar p-\tilde p_\phi)^2,\,\rho_{r,\beta}(x,y)\rho_{r,\beta}(x,y^\prime)\bigr)\Bigr],\\
\Theta(\mu;r,\beta) &= \E_\nu\Bigl[\Cov_{\mu_x^{\otimes 3}}\bigl(\mathbf{1}_{\mathcal{CC}_\infty^{2\delta}},\,\rho_{r,\beta}(x,y_1)\rho_{r,\beta}(x,y_2)\rho_{r,\beta}(x,y_3)\bigr)\Bigr]\,.
\end{align*}

\begin{restatable}{prop}{jointcov}\label{prop:joint-cov}
With $\phi=\phi_d^{(t)}$ and $r_t$ measurable, $e_d(\mu^{(t+1)},\phi)-e_d(\mu^{(t)},\phi) =\Gamma_t\,,P_\infty^{2\delta}(\mu^{(t+1)})-P_\infty^{2\delta}(\mu^{(t)}) =\Theta_t$.
\end{restatable}
Proposition~\ref{prop:joint-cov} says the two quantities driving our static lower bound from Lemma~\ref{lem:condorcetlb} and Lemma~\ref{lem:condorcetgrowth}, namely the embedding error $e_d$ and the cycle probability $P_\infty^{2\delta}$, both shift across rounds via reward-tilt covariances of the same form, differing only in the order of the $\rho$-product. Round-by-round optimization fits $r_t$ to $\gD_t$ but does not control the sign of the resulting covariances $\Gamma_t,\Theta_t$, so neither sequence is guaranteed to decrease, and the joint update has no built-in mechanism for shrinking the lower-bound floor.
\begin{restatable}{prop}{jointfixed}\label{prop:joint-fixed}
At any fixed point $(\mu^\ast,\phi^\ast,r^\ast)$ of the joint update with $\phi^\ast$ optimal on $\mu^\ast$ and $r^\ast\in\mathcal{R}_d(\phi^\ast)$ loss-optimal on $\mu^\ast$, both $\Gamma(\mu^\ast,\phi^\ast;r^\ast,\beta)=0$ and $\Theta(\mu^\ast;r^\ast,\beta)=0$.
\end{restatable}
At a fixed point, invariance of $\mu^\ast$ under the soft-update forces $\rho^\ast=1$ on $\mathrm{supp}(\mu^\ast)$, which means $r^\ast$ is constant on the data and provides no preference signal, and both covariances then vanish trivially. Either the joint dynamics keep moving without converging, or they stabilize only at a degenerate point at which the reward stops ranking responses. In either case, they do not implicitly converge to the sweet spot of the static tradeoff.

\section{Experiments}
We evaluate our theoretical results on both synthetic and real preference data. For the synthetic experiments, we consider a finite item set with a single prompt $x$ and sample unordered pairs uniformly. There are $K$ annotator types. Each type $k$ follows a logistic pairwise model $P_k(y\succ y^\prime\mid x)=F(U_k(y)-U_k(y^\prime))$, and $\bar{p} = \frac{1}{K}\sum_k P_k$ is the arithmetic mean. The utility functions have three components, a transitive ranking component, a cyclic class component that varies across annotator types, and idiosyncratic noise. When cyclic heterogeneity is present, pooling across types yields a non-Bradley-Terry $\bar{p}$. The embedding $\phi_d$ is a random projection into binary codes of length $D_{\max}$, retaining the first $d$ coordinates. This satisfies Assumption~\ref{asmp:filtration} by construction. We report Monte Carlo estimates (for synthetic data) and empirical estimates based on the full dataset (for real datasets) of $\gE_{\mathrm{emb}}(\phi_d)$ and $\gE_{\mathrm{agr}}(\phi_d,r)$, where $r$ is a Bradley-Terry model with one scalar per embedding group (responses sharing the same $\phi_d$ value).  We compare the result for both fitting a explicit reward (which we refer to as RLHF) and for extracting the implicit reward from DPO (which we just refer to as DPO here). More implementation details are given in Appendix~\ref{appendix:exp}.

We then evaluate on three real preference datasets: Jester \cite{eigentaste}, Sushi \cite{kamishima}, and MT-Bench \cite{mtbench}. Jester contains individual ratings of jokes on a continuous scale; we construct pairwise preferences by comparing ratings. Sushi provides ranked preferences over sushi types from multiple users, which we convert to pairwise win probabilities. MT-Bench consists of multi-turn dialogue preferences from human evaluations. For each dataset we embed items, namely joke captions, sushi types, or dialogue responses, using a pretrained sentence transformer.

Figure~\ref{fig:minilm_main} shows the excess loss decomposition as a function of embedding dimension for both RLHF and DPO with the pretrained embedding \texttt{all-MiniLM-L6-v2} \citep{reimers-2019-sentence-bert}. The top row gives RLHF on synthetic, Jester, and Sushi, and the bottom row gives the corresponding DPO results. In each panel, the embedding term decreases with dimension while the agreement term increases, consistent with Proposition~\ref{prop:ed} and Theorem~\ref{thm:cyclebound}. The total excess loss often has a minimum at an intermediate dimension. We provide more experimental results with pretrained embedding \texttt{mpnet} \citep{song2020mpnet} and with the MT-Bench dataset in Appendix~\ref{appendix:exp}.
\begin{figure}[h]
    \centering
    \includegraphics[width=0.32\linewidth]{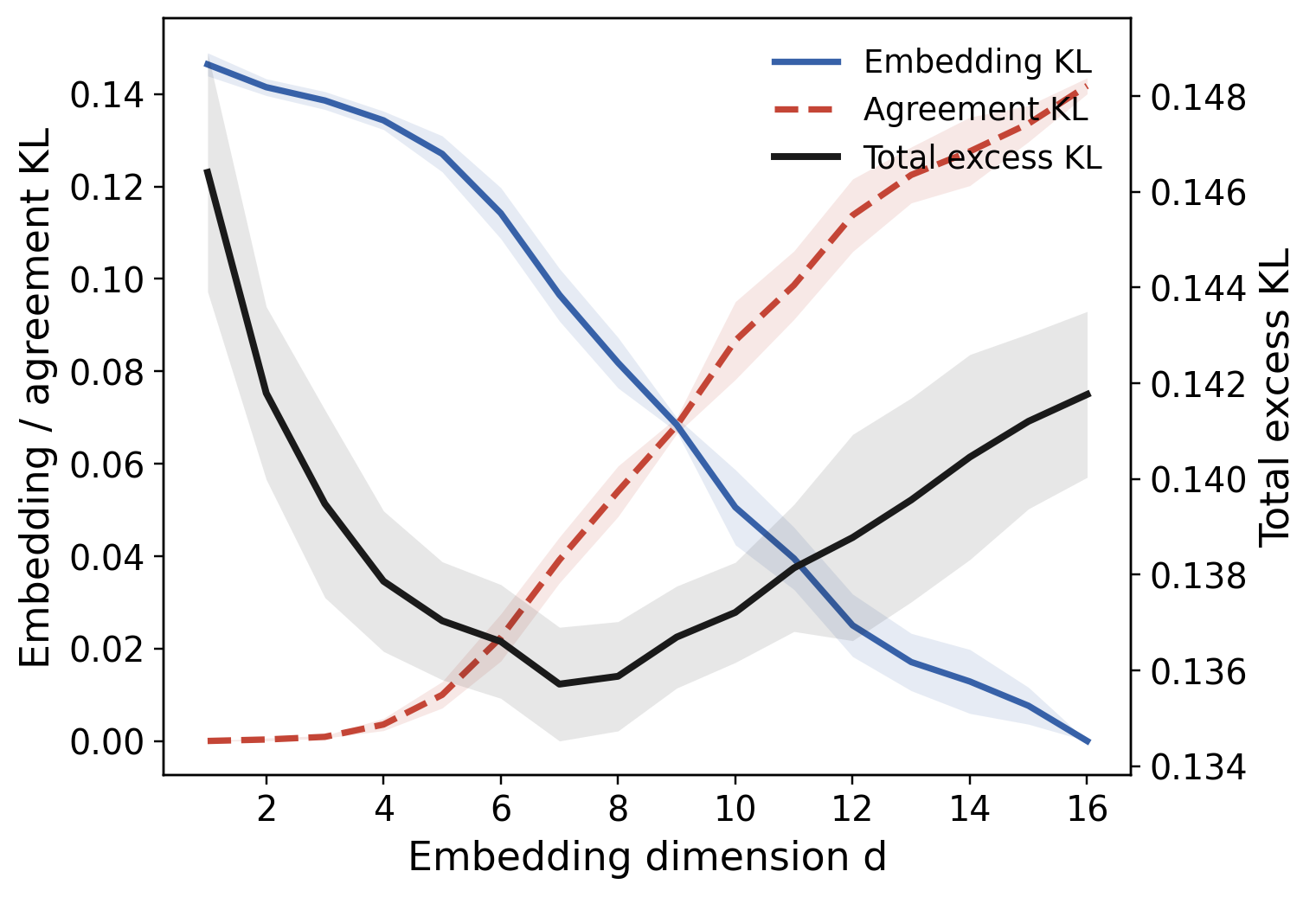}
    \includegraphics[width=0.32\linewidth]{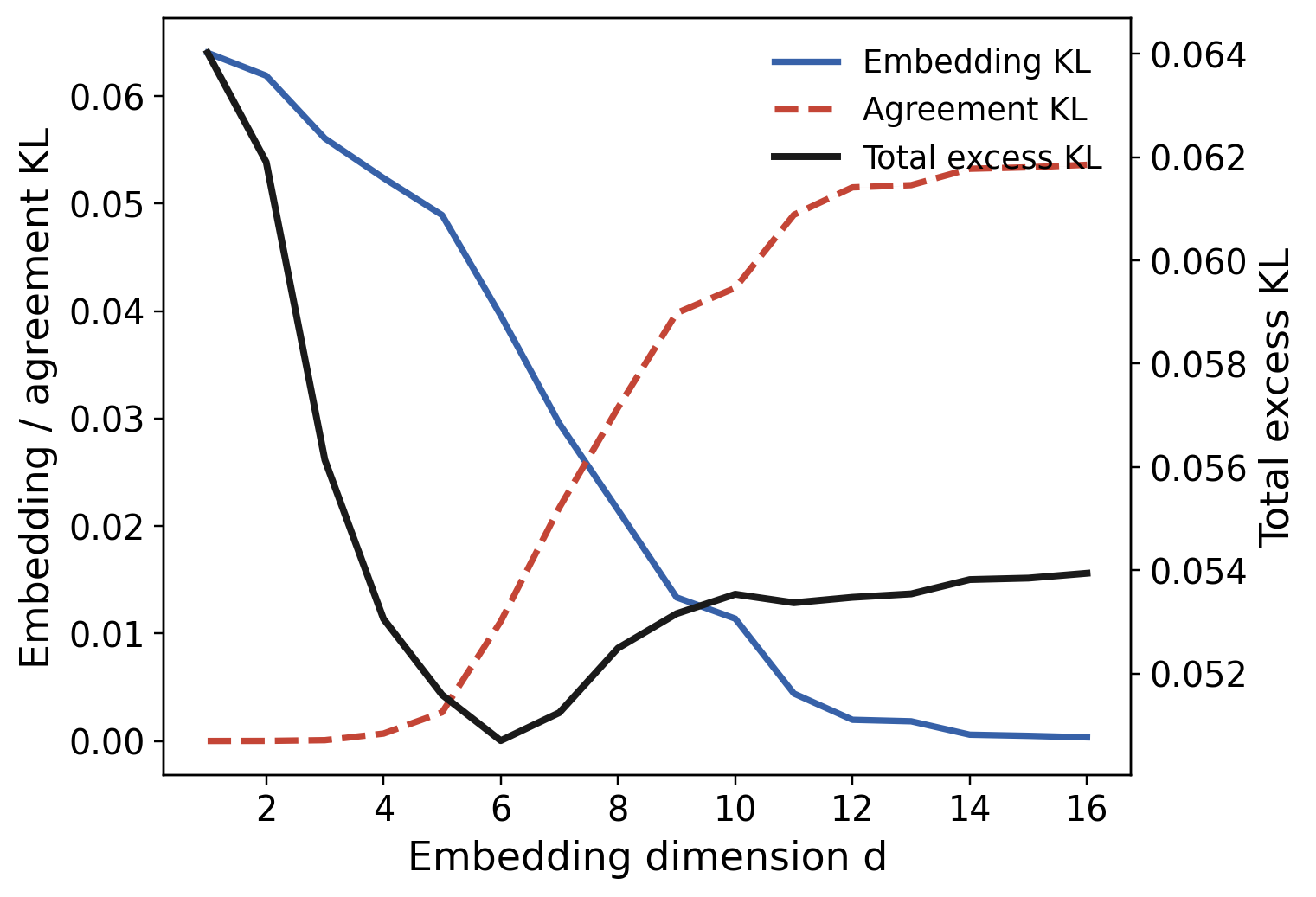}
    \includegraphics[width=0.32\linewidth]{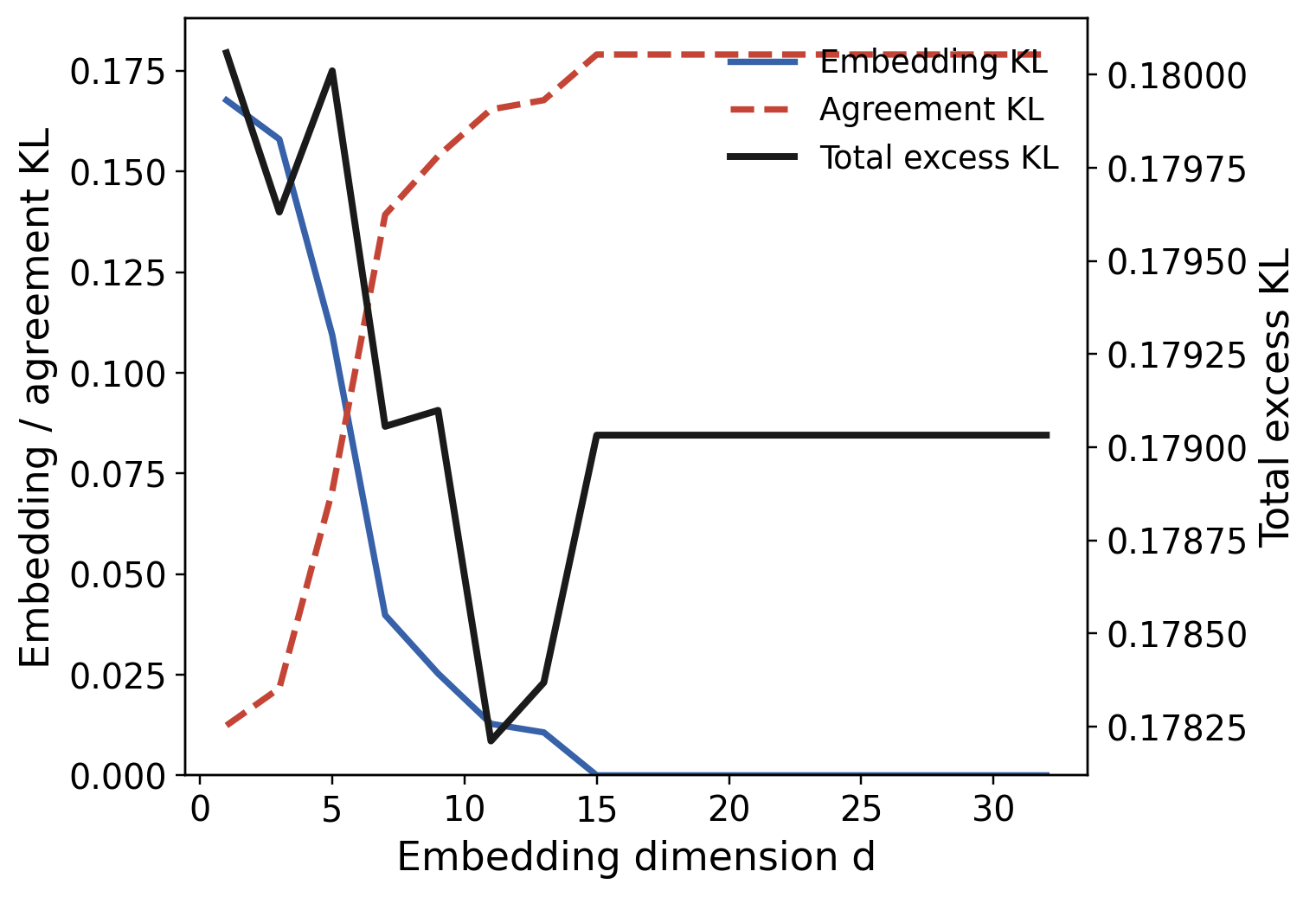}
    \includegraphics[width=0.32\linewidth]{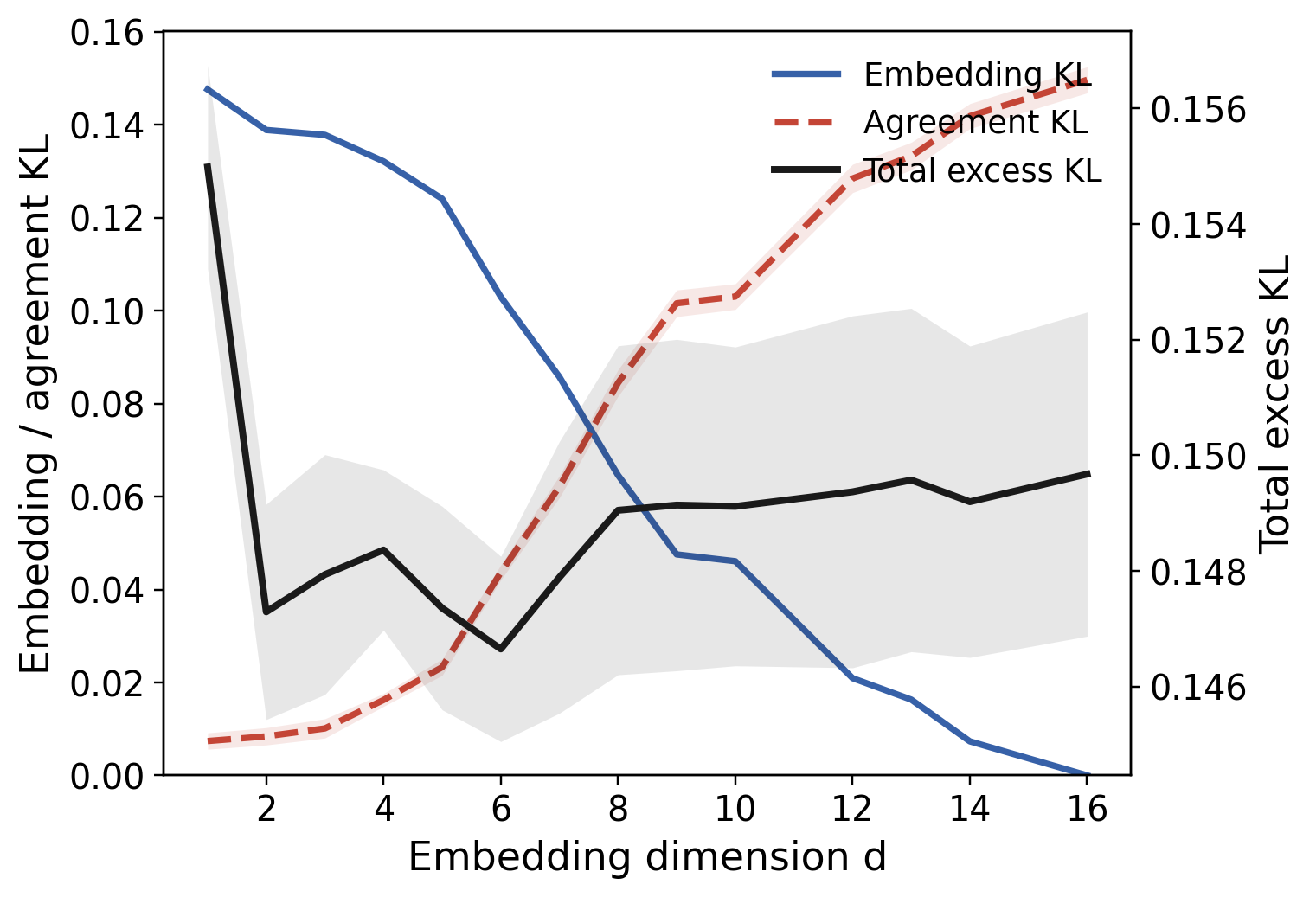}
    \includegraphics[width=0.32\linewidth]{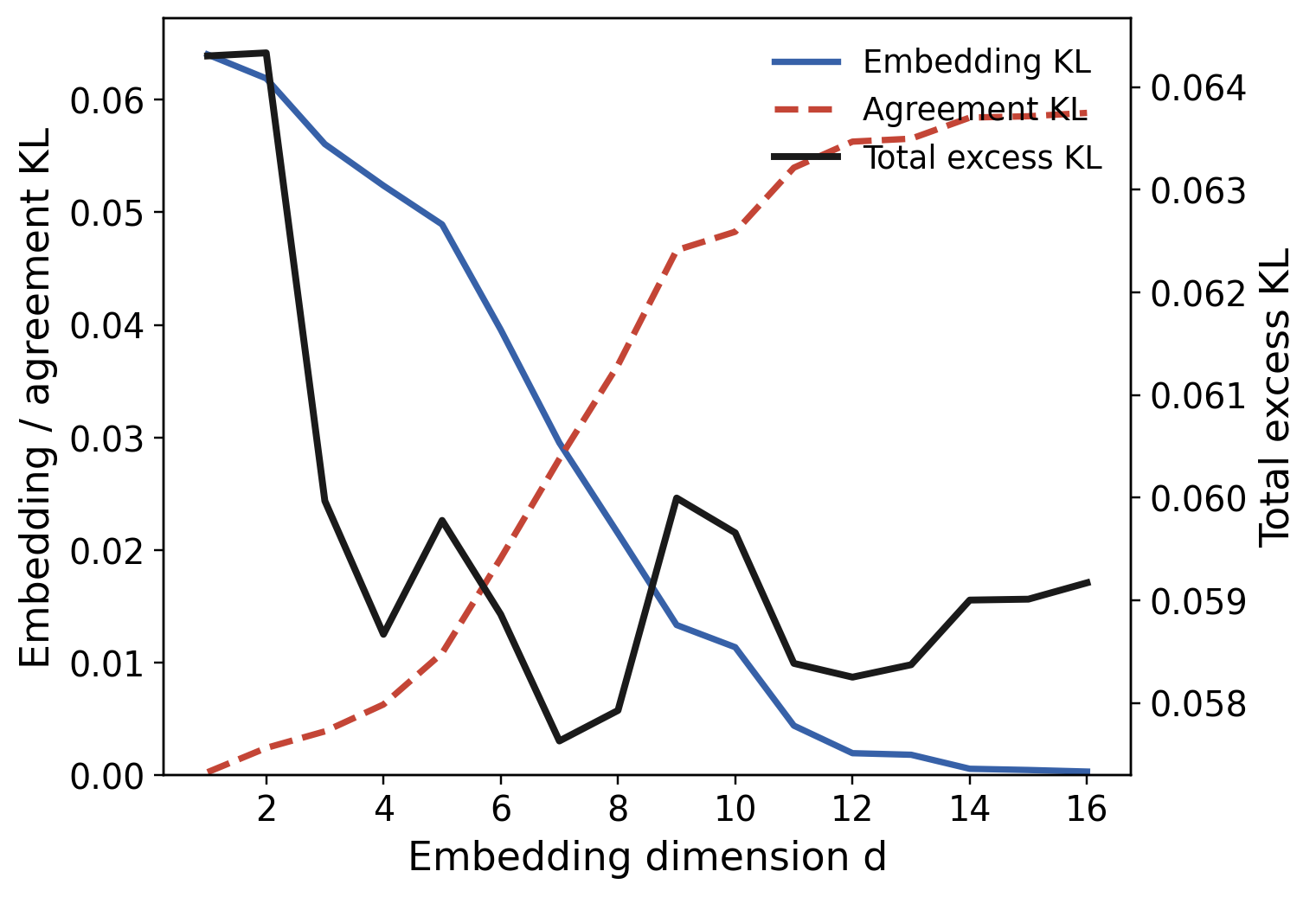}
    \includegraphics[width=0.32\linewidth]{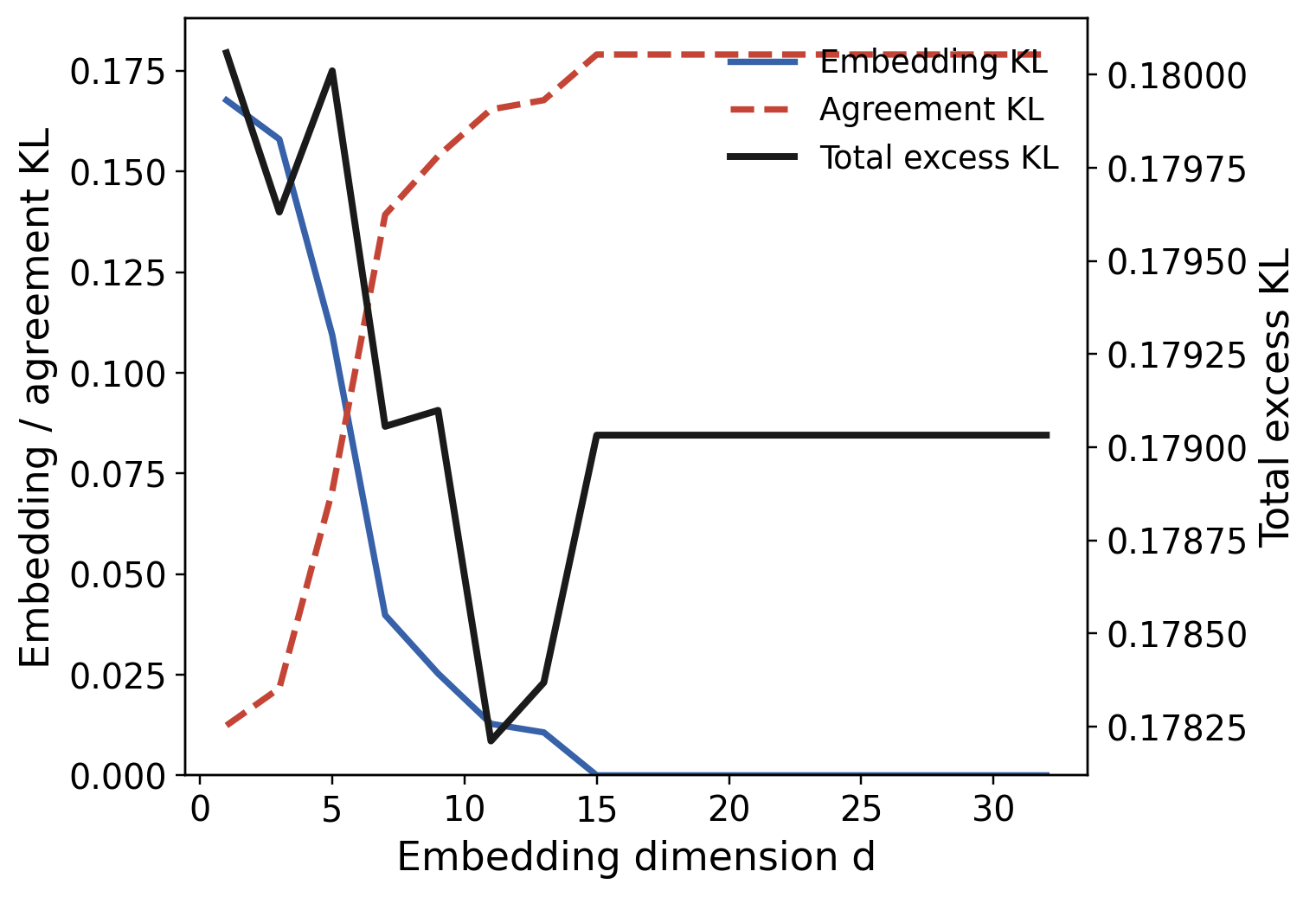}
    \caption{Top row: RLHF on synthetic (left), Jester (middle), and Sushi (right). Bottom row: DPO on synthetic (left), Jester (middle), and Sushi (right).}
    \label{fig:minilm_main}
\end{figure}

Figure~\ref{fig:visualization} provides a geometric visualization of this tradeoff. The top row shows the top 1, 2, and 3 principal components of the embedding used for reward model training on the synthetic data. The dot intensity represents how many data points are mapped to each representation, and the red area highlights the region where points form a Condorcet cycle (Definition~\ref{def:condorcet}). As the embedding dimension increases, it provides a finer-grained representation of the points, and the Condorcet cycle region grows correspondingly.

\begin{figure}[tbh]
    \centering
    \includegraphics[width=\linewidth]{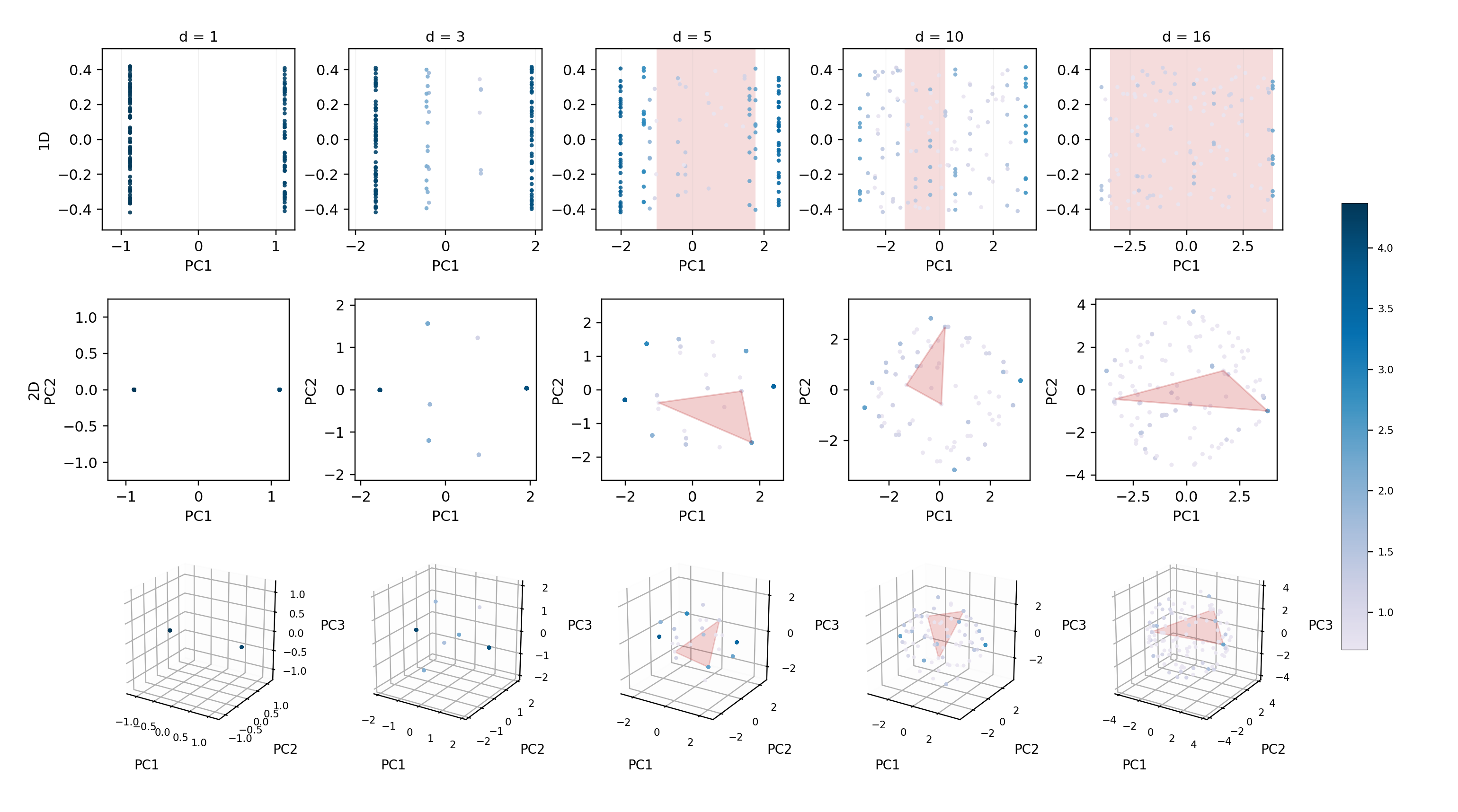}
    \caption{Visualization for the RLHF on the synethtic data. Top 1, 2, and 3 principal components of the embedding used for reward model training on the synthetic data. The intensity of the dots represents how many data points are mapped to each representation, and the red area highlights the region where points form a Condorcet cycle (Definition~\ref{def:condorcet}).} 
    \label{fig:visualization}
\end{figure}

\section{Conclusion}
We show that in RLHF, representation is part of the aggregation problem, not just a modeling detail. The learned embedding determines which preference distinctions are visible to the reward model, inducing a fundamental representation-rationalizability tradeoff. We prove a lower bound on reward excess loss to quantify this tradeoff, where richer representations reduce embedding loss but can increase aggregation inconsistency through cycle structure. Our result extends to DPO and to joint embedding--reward training. In particular, joint updates do not automatically balance  the tradeoff, and can degrade as training proceeds. Experiments on synthetic and real preference datasets confirm our theoretical findings.

\paragraph{Broader Impact}
Our work is mostly theoretical and does not introduce new models, datasets, or deployable systems. To the extent it has societal implications, it clarifies a fundamental limitation of scalar reward models under heterogeneous preferences and may lead to better reward-model design. We do not foresee any direct negative societal impact from the analysis itself.

\section*{Acknowledgments}

Jing Dong is supported by a Distinguished Post Doctoral Fellowship from the Vector Institute.  We acknowledge funding from the Canada CIFAR AI Chair program, discovery grants to Poupart and Yu from the Natural Sciences and Engineering Research Council of Canada and a grant from IITP \& MSIT of Korea (No. RS-2024-00457882, AI Research Hub Project). Computational resources used in preparing this research were provided, in part, by the Province of Ontario, the Government of Canada through CIFAR, and companies sponsoring the Vector Institute https://vectorinstitute.ai/about/current-partners/.

\bibliography{ref}

@inproceedings{Ouyangetal22,
  title     = {Training language models to follow instructions with human feedback},
  author    = {Long Ouyang and Jeffrey Wu and Xu Jiang and Diogo Almeida and Carroll Wainwright and Pamela Mishkin and Chong Zhang and Sandhini Agarwal and Katarina Slama and Alex Gray and John Schulman and Jacob Hilton and Fraser Kelton and Luke Miller and Maddie Simens and Amanda Askell and Peter Welinder and Paul Christiano and Jan Leike and Ryan Lowe},
  booktitle = {Advances in Neural Information Processing Systems},
  year      = {2022},
}

@article{bradley1952rank,
  title={Rank analysis of incomplete block designs: I. the method of paired comparisons},
  author={Bradley, Ralph Allan and Terry, Milton E},
  journal={Biometrika},
  volume={39},
  number={3/4},
  pages={324--345},
  year={1952},
  publisher={JSTOR}
}

@article{plackett1975analysis,
  title={The analysis of permutations},
  author={Plackett, Robin L},
  journal={Journal of the Royal Statistical Society Series C: Applied Statistics},
  volume={24},
  number={2},
  pages={193--202},
  year={1975},
  publisher={Oxford University Press}
}

@book{luce1959individual,
  title={Individual choice behavior: A theoretical analysis},
  author={Luce, R. Duncan},
  year={1959},
  publisher={Wiley}
}

@article{mcfadden1973conditional,
  title={Conditional Logit Analysis of Qualitative Choice Behavior},
  author={Mcfadden, D},
  journal={Frontier in Econometrics},
  year={1973},
  publisher={Academic Press}
}

@article{song2020mpnet,
  title={{MPN}et: Masked and permuted pre-training for language understanding},
  author={Song, Kaitao and Tan, Xu and Qin, Tao and Lu, Jianfeng and Liu, Tie-Yan},
  journal={Advances in Neural Information Processing Systems},
  volume={33},
  pages={16857--16867},
  year={2020}
}

@inproceedings{reimers-2019-sentence-bert,
  title = "Sentence-{BERT}: Sentence Embeddings using Siamese {BERT}-Networks",
  author = "Reimers, Nils and Gurevych, Iryna",
  booktitle = "Proceedings of the Conference on Empirical Methods in Natural Language Processing",
  year = "2019",
}

@article{eigentaste,
  author  = {Goldberg, Ken and Roeder, Theresa and Gupta, Dhruv and Perkins, Chris},
  title   = {Eigentaste: A Constant Time Collaborative Filtering Algorithm},
  journal = {Information Retrieval},
  year    = {2001},
  volume  = {4},
  number  = {2},
  pages   = {133--151}
}

@article{gibbard1973manipulation,
  title={Manipulation of voting schemes: a general result},
  author={Gibbard, Allan},
  journal={Econometrica},
  pages={587--601},
  year={1973},
  publisher={JSTOR}
}

@inproceedings{conitzerposition24,
  title={Position: Social Choice Should Guide {AI} Alignment in Dealing with Diverse Human Feedback},
  author={Conitzer, Vincent and Freedman, Rachel and Heitzig, Jobst and Holliday, Wesley H and Jacobs, Bob M and Lambert, Nathan and Moss{\'e}, Milan and Pacuit, Eric and Russell, Stuart and Schoelkopf, Hailey and Emanuel Tewolde and William S. Zwicker},
  booktitle={Forty-first International Conference on Machine Learning},
  year={2024}
}

@article{satterthwaite1975strategy,
  title={Strategy-proofness and {A}rrow's conditions: Existence and correspondence theorems for voting procedures and social welfare functions},
  author={Satterthwaite, Mark Allen},
  journal={Journal of economic theory},
  volume={10},
  number={2},
  pages={187--217},
  year={1975},
  publisher={Elsevier}
}

@inproceedings{azar2024general,
  title={A general theoretical paradigm to understand learning from human feedback},
  author={Azar, Mohammad Gheshlaghi and Rowland, Mark and Piot, Bilal and Guo, Daniel and Calandriello, Daniele and Valko, Michal and Munos, R{\'e}mi},
  booktitle={International Conference on Artificial Intelligence and Statistics (AISTATS)},
  year={2024}
}

@inproceedings{costereward,
  title={Reward Model Ensembles Help Mitigate Overoptimization},
  author={Coste, Thomas and Anwar, Usman and Kirk, Robert and Krueger, David},
  booktitle={The Twelfth International Conference on Learning Representations},
  year={2024}
}

@article{yang2024regularizing,
  title={Regularizing hidden states enables learning generalizable reward model for {LLM}s},
  author={Yang, Rui and Ding, Ruomeng and Lin, Yong and Zhang, Huan and Zhang, Tong},
  journal={Advances in Neural Information Processing Systems},
  volume={37},
  pages={62279--62309},
  year={2024}
}

@article{rafailov2024scaling,
  title={Scaling laws for reward model overoptimization in direct alignment algorithms},
  author={Rafailov, Rafael and Chittepu, Yaswanth and Park, Ryan and Sikchi, Harshit Sushil and Hejna, Joey and Knox, Brad and Finn, Chelsea and Niekum, Scott},
  journal={Advances in Neural Information Processing Systems},
  volume={37},
  pages={126207--126242},
  year={2024}
}

@inproceedings{munos2024nash,
  title={Nash learning from human feedback},
  author={Munos, R{\'e}mi and Valko, Michal and Calandriello, Daniele and Azar, Mohammad Gheshlaghi and Rowland, Mark and Guo, Zhaohan Daniel and Tang, Yunhao and Geist, Matthieu and Mesnard, Thomas and Fiegel, C{\^o}me and Andrea Michi and  Marco Selvi and Sertan Girgin and Nikola Momchev and Olivier Bachem and Daniel J. Mankowitz and Doina Precup and Bilal Piot},
  booktitle={Forty-first International Conference on Machine Learning},
  year={2024}
}

@article{kamishima,
  author  = {Kamishima, Toshihiro},
  title   = {Nantonac Collaborative Filtering: Recommendation Based on Order Responses},
  journal = {Knowledge Discovery and Data Mining},
  year    = {2003},
  pages   = {573--579}
}

@inproceedings{zhang2025beyond,
  title={Beyond {B}radley-{T}erry Models: A General Preference Model for Language Model Alignment},
  author={Zhang, Yifan and Zhang, Ge and Wu, Yue and Xu, Kangping and Gu, Quanquan},
  booktitle={International Conference on Machine Learning},
  pages={76939--76965},
  year={2025},
}

@inproceedings{zhu2023principled,
  title={Principled reinforcement learning with human feedback from pairwise or k-wise comparisons},
  author={Zhu, Banghua and Jordan, Michael and Jiao, Jiantao},
  booktitle={International Conference on Machine Learning},
  pages={43037--43067},
  year={2023},
}

@unpublished{sun2024rethinking,
  title={Rethinking {B}radley-{T}erry models in preference-based reward modeling: Foundations, theory, and alternatives},
  author={Sun, Hao and Shen, Yunyi and Ton, Jean-Francois},
  note={arXiv preprint arXiv:2411.04991},
  year={2024}
}

@inproceedings{gao2023scaling,
  title={Scaling laws for reward model overoptimization},
  author={Gao, Leo and Schulman, John and Hilton, Jacob},
  booktitle={International Conference on Machine Learning},
  pages={10835--10866},
  year={2023},
}

@inproceedings{zhu2024iterative,
  title={Iterative data smoothing: mitigating reward overfitting and overoptimization in {RLHF}},
  author={Zhu, Banghua and Jordan, Michael I and Jiao, Jiantao},
  booktitle={Proceedings of the 41st International Conference on Machine Learning},
  pages={62405--62428},
  year={2024}
}

@article{casper2023open,
  title={Open Problems and Fundamental Limitations of Reinforcement Learning from Human Feedback},
  author={Casper, Stephen and Davies, Xander and Shi, Claudia and Krendl Gilbert, Thomas and Scheurer, J{\'e}r{\'e}my and Rando Ramirez, Javier and Freedman, Rachel and Korbak, Tomasz and Lindner, David and Freire, Pedro and Tony Wang and Samuel Marks and Charbel-Rapha{\"e}l Segerie and Micah Carroll and Andi Peng and Phillip Christoffersen and Mehul Damani and Stewart Slocum and Usman Anwar and Anand Siththaranjan and Max Nadeau and Eric J. Michaud and Jacob Pfau and Dmitrii Krasheninnikov and Xin Chen and Lauro Langosco and Peter Hase and Erdem Bıyık and Anca Dragan and David Krueger and Dorsa Sadigh and Dylan Hadfield-Menell},
  journal={Transactions on Machine Learning Research},
  year={2023},
  publisher={OpenReview}
}

@unpublished{wang2024secrets,
  title={Secrets of {RLHF} in large language models part {II}: Reward modeling},
  author={Wang, Binghai and Zheng, Rui and Chen, Lu and Liu, Yan and Dou, Shihan and Huang, Caishuang and Shen, Wei and Jin, Senjie and Zhou, Enyu and Shi, Chenyu and Songyang Gao and Nuo Xu and Yuhao Zhou and Xiaoran Fan and Zhiheng Xi and Jun Zhao and Xiao Wang and Tao Ji and Hang Yan and Lixing Shen and Zhan Chen and Tao Gui and Qi Zhang and Xipeng Qiu and Xuanjing Huang and Zuxuan Wu and Yu-Gang Jiang},
  note={arXiv preprint arXiv:2401.06080},
  year={2024}
}

@book{arrow2012social,
  title={Social choice and individual values},
  author={Arrow, Kenneth J.},
  edition = {2nd},
  year={1963},
  publisher={Yale University Press}
}

@article{mtbench,
  author  = {Zheng, Lianmin and Chiang, Wei-Lin and Sheng, Ying and Zhuang, Siyuan and Wu, Zhanghao and Zhuang, Yonghao and Lin, Zi and Li, Zhuohan and Li, Dacheng and Xing, Eric P. and Zhang, Hao and Gonzalez, Joseph E. and Stoica, Ion},
  title   = {Judging {LLM}-as-a-{J}udge with {MT}-{B}ench and {C}hatbot {A}rena},
  journal = {Advances in Neural Information Processing Systems},
  year    = {2023},
  volume  = {36}
}

@book{de2014essai,
  title={Essai sur l'application de l'analyse {\`a} la probabilit{\'e} des d{\'e}cisions rendues {\`a} la pluralit{\'e} des voix},
  author={Nicolas de Condorcet},
  year={1784},
  publisher={reprinted by Cambridge University Press}
}

@article{hornik1989multilayer,
  title={Multilayer feedforward networks are universal approximators},
  author={Hornik, Kurt and Stinchcombe, Maxwell and White, Halbert},
  journal={Neural networks},
  volume={2},
  number={5},
  pages={359--366},
  year={1989},
  publisher={Elsevier}
}

@inproceedings{rafailov2023direct,
  title={Direct preference optimization: Your language model is secretly a reward model},
  author={Rafailov, Rafael and Sharma, Archit and Mitchell, Eric and Manning, Christopher D and Ermon, Stefano and Finn, Chelsea},
  booktitle={Advances in Neural Information Processing Systems},
  pages={53728--53741},
  year={2023}
}

@article{ge2024axioms,
  title={Axioms for {AI} alignment from human feedback},
  author={Ge, Luise and Halpern, Daniel and Micha, Evi and Procaccia, Ariel D and Shapira, Itai and Vorobeychik, Yevgeniy and Wu, Junlin},
  journal={Advances in Neural Information Processing Systems},
  volume={37},
  pages={80439--80465},
  year={2024}
}

@inproceedings{hollenderenforcing,
  title={Enforcing Axioms for {AI} Alignment under Loss-Based Rules},
  author={Hollender, Alexandros and Kraiczy, Sonja},
  booktitle={The Fourteenth International Conference on Learning Representations},
  year={2026}
}

@unpublished{liu2025statistical,
  title={Statistical impossibility and possibility of aligning {LLM}s with human preferences: From {C}ondorcet paradox to {N}ash equilibrium},
  author={Liu, Kaizhao and Long, Qi and Shi, Zhekun and Su, Weijie J and Xiao, Jiancong},
  note={arXiv preprint arXiv:2503.10990},
  year={2025}
}

@unpublished{xiao2025theoretical,
  title={Theoretical tensions in {RLHF}: Reconciling empirical success with inconsistencies in social choice theory},
  author={Xiao, Jiancong and Shi, Zhekun and Liu, Kaizhao and Long, Qi and Su, Weijie J},
  note={arXiv preprint arXiv:2506.12350},
  year={2025}
}

@inproceedings{siththaranjandistributional24,
  title={Distributional Preference Learning: Understanding and Accounting for Hidden Context in {RLHF}},
  author={Siththaranjan, Anand and Laidlaw, Cassidy and Hadfield-Menell, Dylan},
  booktitle={The Twelfth International Conference on Learning Representations},
  year={2024}
}

@inproceedings{dai2024mapping,
  title={Mapping Social Choice Theory to {RLHF}},
  author={Dai, Jessica and Fleisig, Eve},
  booktitle={ICLR 2024 Workshop on Reliable and Responsible Foundation Models},
  year={2024}
}

\newpage
\appendix
\section{Proofs for Section~\ref{sec:tradeoff}}\label{appendix:tradeoff}
\decomposition*
\begin{proof}
For fixed $(x, y, y^\prime)$, minimizing $h(q)=-p\log q-(1-p)\log(1-q)$ over $q\in(0,1)$ with $p=\bar{p}(x,y, y^\prime)$ gives the minimizer $q=p$, so the best pointwise fit satisfies $F(\Delta r(x, y, y^\prime))=\bar{p}(x,y, y^\prime)$. Plugging $q=p=\bar{p}$ gives the Bayes-optimal loss
\begin{align*}
  \gL^\ast = \E_{x\sim\nu}\,\E_{(y,y^\prime)\sim\mu_x^{\otimes 2}}\left[ H\left(\bar{p}(x,y, y^\prime)\right) \right],
  \qquad H(p)=-p\log p-(1-p)\log(1-p).
\end{align*}

Moreover,
\begin{align*}
  & \gL(r)-\gL^\ast \nonumber\\
  &= \E_{\substack{x\sim\nu \\ (y,y^\prime)\sim\mu_x^{\otimes 2}}} \left[
    \bar{p}(x,y, y^\prime)\log\frac{\bar{p}(x,y, y^\prime)}{F(\Delta r(x,y, y^\prime))}
    +\left(1-\bar{p}(x,y, y^\prime)\right)\log\frac{1-\bar{p}(x,y, y^\prime)}{1-F(\Delta r(x,y, y^\prime))}
  \right] \nonumber \\
  &= \E_{\substack{x\sim\nu \\ (y,y^\prime)\sim\mu_x^{\otimes 2}}}\left[
    \KL\left(\bar{p}(x,y, y^\prime)\,\big\|\,F(\Delta r(x,y, y^\prime))\right)
  \right].
\end{align*}

Define the embedding-induced win probability as
\begin{align*}
  \tilde{p}_\phi(x, y, y^\prime) = \E_{(y_0,y_0')\sim\mu_x^{\otimes 2}} \left[
    \bar{p}(x, y_0,y_0')\mid\phi(x,y_0)=\phi(x,y),\,\phi(x,y_0')=\phi(x,y^\prime)
  \right].
\end{align*}
When $\phi$ is injective, $\tilde{p}_\phi=\bar{p}$. Otherwise $\tilde{p}_\phi$ averages $\bar{p}$ over pairs that $\phi$ cannot distinguish.

To decompose $\KL$, we use
\begin{align*}
  \bar{p}\log\frac{\bar{p}}{F(\Delta r)}
  &=\bar{p}\log\frac{\bar{p}}{\tilde{p}_\phi}+\bar{p}\log\frac{\tilde{p}_\phi}{F(\Delta r)}, \\
  (1-\bar{p})\log\frac{1-\bar{p}}{1-F(\Delta r)}
  &=(1-\bar{p})\log\frac{1-\bar{p}}{1-\tilde{p}_\phi}+(1-\bar{p})\log\frac{1-\tilde{p}_\phi}{1-F(\Delta r)},
\end{align*}
where $\bar{p}=\bar{p}(x,y, y^\prime)$, $\tilde{p}_\phi=\tilde{p}_\phi(x,y, y^\prime)$, and $F(\Delta r)=F(\Delta r(x,y, y^\prime))$ at each $(x,y, y^\prime)$.
Thus
\begin{align*}
  & \E_{x\sim\nu}\E_{(y, y^\prime)\sim\mu_x^{\otimes 2}}\bigl[\KL\bigl(\bar{p}(x,y, y^\prime)\big\|\,F(\Delta r(x,y, y^\prime))\bigr)\bigr] \\
  &= \E_{x\sim\nu}\E_{(y, y^\prime)\sim\mu_x^{\otimes 2}}\bigl[\KL\bigl(\bar{p}(x,y, y^\prime)\big\|\tilde{p}_\phi(x,y, y^\prime)\bigr)\bigr] \\
  &\quad + \E_{x\sim\nu}\E_{(y, y^\prime)\sim\mu_x^{\otimes 2}}\Biggl[
      \bar{p}\log\frac{\tilde{p}_\phi}{F(\Delta r)} + (1-\bar{p})\log\frac{1-\tilde{p}_\phi}{1-F(\Delta r)}
    \Biggr],
\end{align*}

Let $(Y,Y^\prime)\sim\mu_x^{\otimes 2}$ denote the underlying random variables on $\gX\times\gY^2$ and set $\gG=\sigma(\phi(x,Y),\phi(x,Y^\prime))$.

Then $\tilde{p}_\phi(x,Y, Y^\prime)=\E\left[\bar{p}(x,Y, Y^\prime)\mid \gG\right]$ by the equivalent definition of $\tilde{p}_\phi$ above, hence $\E\left[1-\bar{p}(x,Y, Y^\prime)\mid \gG\right]=1-\tilde{p}_\phi(x,Y, Y^\prime)$.
If $r(x,y)=g(\phi(x,y))$, then $F(\Delta r(x,Y, Y^\prime))$ is a measurable function of $(\phi(x,Y),\phi(x,Y^\prime))$, so $F(\Delta r(x,Y, Y^\prime))$ is $\gG$-measurable.

Using the tower property, we have
\begin{align*}
  \E\left[\bar{p}(x,y, y^\prime)\log\frac{\tilde{p}_\phi(x,y, y^\prime)}{F(\Delta r(x,y, y^\prime))}\right]
  &=\E\left[\E\left[\bar{p}(x,y, y^\prime)\mid\gG\right]\log\frac{\tilde{p}_\phi(x,y, y^\prime)}{F(\Delta r(x,y, y^\prime))}\right] \\
  &=\E\left[\tilde{p}_\phi(x,y, y^\prime)\log\frac{\tilde{p}_\phi(x,y, y^\prime)}{F(\Delta r(x,y, y^\prime))}\right], \\
  \E\left[(1-\bar{p}(x,y, y^\prime))\log\frac{1-\tilde{p}_\phi(x,y, y^\prime)}{1-F(\Delta r(x,y, y^\prime))}\right]
  &=\E\left[\E\left[1-\bar{p}(x,y, y^\prime)\mid\gG\right]\log\frac{1-\tilde{p}_\phi(x,y, y^\prime)}{1-F(\Delta r(x,y, y^\prime))}\right] \\
  &=\E\left[(1-\tilde{p}_\phi(x,y, y^\prime))\log\frac{1-\tilde{p}_\phi(x,y, y^\prime)}{1-F(\Delta r(x,y, y^\prime))}\right].
\end{align*}

Therefore,
\begin{align*}
  \gL(r)-\gL^\ast
  &= \E_{x\sim\nu}\,\E_{(y, y^\prime)\sim\mu_x^{\otimes 2}}\left[
    \KL\left(\bar{p}(x,y, y^\prime)\,\big\|\,\tilde{p}_\phi(x,y, y^\prime)\right)
  \right] \nonumber \\
  &\quad + \E_{x\sim\nu}\,\E_{(y, y^\prime)\sim\mu_x^{\otimes 2}}\left[
    \KL\left(\tilde{p}_\phi(x,y, y^\prime)\,\big\|\,F(\Delta r(x,y, y^\prime))\right)
  \right].
\end{align*}
\end{proof}

\subsection{Proofs for SubSection~\ref{subsec:representation}}\label{appendix:representation}
\monoed*
\begin{proof}
Let $X=\bar{p}(x,y, y^\prime)$ under $x\sim\nu$, $(y, y^\prime)\sim\mu_x^{\otimes 2}$. Then $\E\left[X\mid\gG_d\right]=\tilde{p}_{\phi_d}(x,y, y^\prime)$ a.s., hence $e_d=\left\|X-\E\left[X\mid\gG_d\right]\right\|_2^2$.

By Assumption~\ref{asmp:filtration}, $\gG_d\subseteq\gG_{d+1}$, thus the tower property gives
$\E\left[X\mid\gG_d\right]=\E\left[\E\left[X\mid\gG_{d+1}\right]\mid\gG_d\right]$, and
\begin{align*}
  \left\|X-\E\left[X\mid\gG_d\right]\right\|_2^2
  = \ & 
  \left\|X-\E\left[X\mid\gG_{d+1}\right]\right\|_2^2
  +\left\|\E\left[X\mid\gG_{d+1}\right]-\E\left[X\mid\gG_d\right]\right\|_2^2 \\
  \geq \ & 
  \left\|X-\E\left[X\mid\gG_{d+1}\right]\right\|_2^2. 
\end{align*}
\end{proof}

\newpage
\subsection{Proofs for SubSection~\ref{subsec:cycle}}
\cyclebound*
\begin{proof}
For each $(x,y,y^\prime)$ set $\eta=\eta_d(x,y,y^\prime)\in(0,\frac{1}{2}]$, $\kappa=\kappa_d(x,y,y^\prime)=\eta(1-\eta)$, $p=\tilde p_{\phi_d}(x,y,y^\prime)$, and
\[
I_{d,(x,y,y^\prime)}=\left[\log\tfrac{\eta}{1-\eta},\ \log\tfrac{1-\eta}{\eta}\right].
\]
The map $f_p(t)=\KL(p\|F(t))$ satisfies $f_p'(t)=F(t)-p$ and $f_p''(t)=F(t)(1-F(t))$. On $I_{d,(x,y,y^\prime)}$, $F(t)\in[\eta,1-\eta]$, so $f_p''(t)\ge\kappa$. The unique minimizer $t_p=F^{-1}(p)$ lies in $I_{d,(x,y,y^\prime)}$ with $f_p(t_p)=0$, and strong convexity gives
\[
  f_p(t)\geq \tfrac{\kappa}{2}\,(t-t_p)^2,\qquad\forall t\in I_{d,(x,y,y^\prime)}.
\]

 Fix $r\in\mathcal{R}_d$ and set $t=\Delta r(x,y,y^\prime)$. By definition of $\mathcal{R}_d$, $F(t)\in[\eta,1-\eta]$, so $t\in I_{d,(x,y,y^\prime)}$. Applying the pointwise bound and integrating,
\begin{equation}\label{eq:cyc-pointwise}
  \gE_{\mathrm{agr}}(\phi_d,r)
  \geq \frac{1}{2}\,\E_{x\sim\nu}\,\E_{(y,y^\prime)\sim\mu_x^{\otimes 2}}
  \Bigl[\kappa_d(x,y,y^\prime)\,
        \bigl(W_d(x,y,y^\prime)-\Delta r(x,y,y^\prime)\bigr)^2\Bigr].
\end{equation}

For $(y_1,y_2,y_3)\sim\mu_x^{\otimes 3}$ write $a_{ij}(x,y_i,y_j)=W_d(x,y_i,y_j)-\Delta r(x,y_i,y_j)$. Since the cyclic sum of $\Delta r$ vanishes,
\[
C_d(x,y_1,y_2,y_3)=W_d(x,y_1,y_2)+W_d(x,y_2,y_3)+W_d(x,y_3,y_1)=a_{12}+a_{23}+a_{31}.
\]

For each triplet, $(a_{12}+a_{23}+a_{31})^2\le 3(a_{12}^2+a_{23}^2+a_{31}^2)$, and $\kappa_d^{\min}\le\kappa_d(x,y_i,y_j)$ for each pair $(i,j)\in\{(1,2),(2,3),(3,1)\}$, so
\begin{align}\label{eq:cyc-min-bound}
  \kappa_d^{\min}(x,y_1,y_2,y_3)\,C_d^2
  &\leq 3\,\kappa_d^{\min}(x,y_1,y_2,y_3)\,(a_{12}^2+a_{23}^2+a_{31}^2) \notag\\
  &\leq 3\bigl(\kappa_d(x,y_1,y_2)\,a_{12}^2+\kappa_d(x,y_2,y_3)\,a_{23}^2+\kappa_d(x,y_3,y_1)\,a_{31}^2\bigr).
\end{align}
Take $\E_{(y_1,y_2,y_3)\sim\mu_x^{\otimes 3}}$ on both sides. Each $a_{ij}^2$ is a function of only two of the three coordinates, so by exchangeability of $\mu_x^{\otimes 3}$ each of the three terms on the right marginalizes to the same pair expectation $\E_{(y,y^\prime)\sim\mu_x^{\otimes 2}}\bigl[\kappa_d(x,y,y^\prime)(W_d-\Delta r)^2\bigr]$. Therefore
\begin{equation}\label{eq:cyc-triplet-to-pair}
  \E_{(y_1,y_2,y_3)\sim\mu_x^{\otimes 3}}\bigl[\kappa_d^{\min}\,C_d^2\bigr]
  \leq 9\,\E_{(y,y^\prime)\sim\mu_x^{\otimes 2}}\bigl[\kappa_d(x,y,y^\prime)\,(W_d-\Delta r)^2\bigr].
\end{equation}

 Taking $\E_x$ in~\eqref{eq:cyc-triplet-to-pair} and combining with~\eqref{eq:cyc-pointwise},
\[
  \gE_{\mathrm{agr}}(\phi_d,r)
  \geq \frac{1}{18}\,
  \E_{x\sim\nu}\,\E_{(y_1,y_2,y_3)\sim\mu_x^{\otimes 3}}
  \Bigl[\kappa_d^{\min}(x,y_1,y_2,y_3)\,C_d(x,y_1,y_2,y_3)^2\Bigr].
\]
Since the right-hand side does not depend on $r\in\mathcal{R}_d$, the bound applies to $\inf_{r\in\mathcal{R}_d}\gE_{\mathrm{agr}}(\phi_d,r)$.
\end{proof}
\condorcetlb*
\begin{proof}
Since $(F^{-1})'(p)=1/(p(1-p))>0$, the map $F^{-1}$ is strictly increasing on $(0,1)$.

On a positive $\delta$-cycle, each term has $\tilde{p}_{\phi_d}\geq\frac{1}{2}+\delta$, so
\begin{align*}
    C_d(x,y_1,y_2,y_3) 
    = \ & F^{-1} \left(\tilde{p}_{\phi_d}(x,y_1,y_2)\right) + F^{-1} \left(\tilde{p}_{\phi_d}(x,y_2,y_3)\right) + F^{-1} \left(\tilde{p}_{\phi_d}(x,y_3,y_1)\right)\\
    \geq \ & 3F^{-1}\left(\frac{1}{2} + \delta\right) = 3\ell_\delta > 0 \,.
\end{align*}
Similarly, on a negative $\delta$-cycle each term is at most $-\ell_\delta$ and $C_d\leq -3\ell_\delta$. Therefore $C_d^2\geq 9\ell_\delta^2$ on $\mathcal{CC}_d^\delta$, hence
\[
  \E_{x,(y_1,y_2,y_3)}\bigl[\kappa_d^{\min}\,C_d^2\bigr]
  \geq \E_{x,(y_1,y_2,y_3)}\bigl[\kappa_d^{\min}\,C_d^2\,\mathbf{1}_{\mathcal{CC}_d^\delta}\bigr]
  \geq 9\ell_\delta^2\,\E_{x,(y_1,y_2,y_3)}\bigl[\kappa_d^{\min}\,\mathbf{1}_{\mathcal{CC}_d^\delta}\bigr]
  = 9\ell_\delta^2\,\bar P_d^\delta.
\]
Combining with Theorem~\ref{thm:cyclebound},
\[
\inf_{r\in\mathcal{R}_d}\gE_{\mathrm{agr}}(\phi_d,r)\geq \tfrac{1}{18}\,\E\bigl[\kappa_d^{\min}\,C_d^2\bigr]\geq \tfrac{9\ell_\delta^2}{18}\,\bar P_d^\delta = \tfrac{\ell_\delta^2}{2}\,\bar P_d^\delta.
\]
\end{proof}
\condorcetgrowth*
\begin{proof}
\textbf{(1)} Under $\phi_0$, conditioning is trivial and $\tilde{p}_{\phi_0}(x,y, y^\prime)=\E_{(y_0,y_0')\sim\mu_x^{\otimes 2}}[\bar{p}(x,y_0,y_0')]$. Symmetry of $\mu_x^{\otimes2}$ in $(y_0,y_0')$ and $\bar{p}(x,y, y^\prime)+\bar{p}(x,y^\prime,y)=1$ imply $\tilde{p}_{\phi_0}=\frac{1}{2}$ everywhere, hence $\mathcal{CC}_0^\delta$ has probability zero and therefore $\bar P_0^\delta=0$.

\medskip\noindent\textbf{(2)}
Let $\mathcal{CC}_\infty^{2\delta}$ be the population analogue of Definition~\ref{def:condorcet} with margin $2\delta$ and $\bar{p}$ in place of $\tilde{p}_{\phi_d}$, so $P_\infty^{2\delta}=\Prob(\mathcal{CC}_\infty^{2\delta})$.

If a triplet is in $\mathcal{CC}_{\infty,+}^{2\delta}$ and each of the three pairs $(y_1,y_2)$, $(y_2,y_3)$, $(y_3,y_1)$ satisfies $|\tilde{p}_{\phi_d}-\bar{p}|\leq\delta$, then $\tilde{p}_{\phi_d}\geq\frac{1}{2}+\delta$ on each pair, so the triplet lies in $\mathcal{CC}_d^\delta$. The same implication holds for $\mathcal{CC}_{\infty,-}^{2\delta}$. Hence every triplet in $\mathcal{CC}_\infty^{2\delta}\setminus\mathcal{CC}_d^\delta$ has at least one of those three pairs with $|\tilde{p}_{\phi_d}-\bar{p}|>\delta$, i.e.
\[
  \mathbf{1}_{\mathcal{CC}_\infty^{2\delta}}
  \leq
  \mathbf{1}_{\mathcal{CC}_d^\delta}
  +\sum_{(i,j)\in\{(1,2),(2,3),(3,1)\}}
  \mathbf{1}_{\{|\tilde{p}_{\phi_d}(y_i,y_j)-\bar{p}(y_i,y_j)|>\delta\}}.
\]
Multiplying by $\kappa_d^{\min}\in[0,1/4]$ and taking expectations,
\[
  \E\left[\kappa_d^{\min}\mathbf{1}_{\mathcal{CC}_\infty^{2\delta}}\right]
  \le
  \bar P_d^\delta
  +3\,\Prob \left(|\tilde{p}_{\phi_d}(y, y^\prime)-\bar{p}(y, y^\prime)|>\delta\right).
\]
By definition of $\bar P_\infty^{\inf,2\delta}$, we have $\E[\kappa_d^{\min}\mathbf{1}_{\mathcal{CC}_\infty^{2\delta}}]\ge \bar P_\infty^{\inf,2\delta}$, and Markov's inequality gives $\Prob(|\tilde{p}_{\phi_d}-\bar{p}|>\delta) \leq \delta^{-2}\|\tilde{p}_{\phi_d}-\bar{p}\|_2^2 \leq e_d/\delta^2$, which proves item (2). 

\medskip\noindent\textbf{(3)}
If $d_1\leq d_2$ then $e_{d_1}\geq e_{d_2}$ by Lemma~\ref{lem:monoed}, so $\bar P_\infty^{\inf,2\delta}-\frac{3}{\delta^2}e_{d_1} \leq \bar P_\infty^{\inf,2\delta}-\frac{3}{\delta^2}e_{d_2}$.
\end{proof}

\newpage
\subsection{Proofs for SubSection~\ref{subsec:holder}}\label{appendix:holder}

\paragraph{Constructing $\varepsilon$-separating embeddings.} For completeness we note that $\varepsilon$-separating embeddings of arbitrarily small $\varepsilon$ can be realized by finite-width ReLU MLPs whenever the response space is compact and the embedding map $y\mapsto e(y)$ is continuous and injective in $y$.

Let $\mathcal{N}_\gY(\epsilon)$ denote the $\epsilon$-covering number of $(\gY,d_\gY)$. A ReLU MLP is a function $f:\R^D\to\R^m$ of the form $f = W_L\circ\sigma\circ W_{L-1}\circ\cdots\circ\sigma\circ W_1$, where $W_\ell$ are affine maps and $\sigma(t)=\max(0,t)$ is applied coordinate-wise.

\begin{thm}[Universal approximation~\cite{hornik1989multilayer}]\label{thm:uat}
Let $K\subset\R^D$ be compact and $g:K\to\R^m$ be continuous.  For any $\eta>0$, there exists a ReLU MLP $f:\R^D\to\R^m$ of finite depth and width such that $\sup_{x\in K}\|f(x)-g(x)\|_\infty\leq\eta$.
\end{thm}

\begin{restatable}{lem}{mlpeps}\label{lem:mlp-eps}
Let $(\gY,d_\gY)$ be compact and $e:\gY\to\R^D$ be continuous and injective. Set $N_\varepsilon=\mathcal{N}_\gY(\varepsilon/4)$. For every $d\ge N_\varepsilon$, there exists a ReLU MLP $f:\R^D\to\R^d$ such that the embedding $y\mapsto f(e(y))$ is $\varepsilon$-separating.
\end{restatable}

\edupper*
\begin{proof}
Fix $(x,y,y^\prime)$.  For any $(y_0,y_0')$ in the conditioning event defining $\tilde p_{\phi_d}(x,y,y^\prime)$, the $\varepsilon$-separating condition gives $d_\gY(y_0,y)\le\varepsilon$ and $d_\gY(y_0',y^\prime)\leq\varepsilon$.

Recalling $\bar p(x,y,y^\prime)=\frac{1}{n}\sum_{i=1}^n F\bigl(\frac{u_i(x,y)-u_i(x,y^\prime)}{\tau}\bigr)$ and $F'(t)=F(t)(1-F(t))\leq\frac{1}{4}$, the mean-value theorem gives
\begin{align*}
  |\bar p(x,y_0,y_0')-\bar p(x,y,y^\prime)|
  \leq \ & \frac{1}{4\tau n}\sum_{i=1}^n\bigl[|u_i(x,y_0)-u_i(x,y)|+|u_i(x,y_0')-u_i(x,y^\prime)|\bigr] \\
  \leq \ & \frac{L_u}{4\tau}\bigl[d_\gY(y_0,y)^\alpha+d_\gY(y_0',y^\prime)^\alpha\bigr]
  \leq \frac{L_u}{2}\,\varepsilon^\alpha.
\end{align*}
Since this holds for every $(y_0,y_0')$ in the conditioning event, the conditional expectation satisfies $|\tilde p_{\phi_d}(x,y,y^\prime)-\bar p(x,y,y^\prime)|\le\frac{L_u}{2}\varepsilon^\alpha$, whence $e_d=\E[(\bar p-\tilde p_{\phi_d})^2]\le\frac{L_u^2}{4}\varepsilon^{2\alpha}$.
\end{proof}

\excess*
\begin{proof}
Combining Lemma~\ref{lem:condorcetlb} with Lemma~\ref{lem:condorcetgrowth} (i.e., $\bar P_d^\delta\ge \bar P_\infty^{\inf,2\delta}-\frac{3}{\delta^2}e_d$) yields
\[
\inf_{r\in\mathcal{R}_d}\gE_{\mathrm{agr}}(\phi_d,r)\ge\tfrac{\ell_\delta^2}{2}\bar P_d^\delta\ge\tfrac{\ell_\delta^2}{2}\Bigl(\bar P_\infty^{\inf,2\delta}-\tfrac{3}{\delta^2}e_d\Bigr).
\]
Together with Proposition~\ref{prop:ed} ($\gE_{\mathrm{emb}}\ge 2 e_d$) and the decomposition $\gL(r)-\gL^\ast=\gE_{\mathrm{emb}}+\gE_{\mathrm{agr}}$,
\begin{align*}
  \gL(r)-\gL^\ast
  &\ge 2\,e_d+\tfrac{\ell_\delta^2}{2}\Bigl(\bar P_\infty^{\inf,2\delta}-\tfrac{3}{\delta^2}e_d\Bigr)\\
  &= (2-\beta)\,e_d+\tfrac{1}{2}\,\ell_\delta^2\,\bar P_\infty^{\inf,2\delta},
\end{align*}
where $\beta=\frac{3\,\ell_\delta^2}{2\,\delta^2}$.

Define the affine function $f(t)=(2-\beta)\,t+\frac{1}{2}\,\ell_\delta^2\,\bar P_\infty^{\inf,2\delta}$ for $t\ge 0$. By Lemma~\ref{lem:ed-upper}, $e_d\in[0,\,\tfrac{L_u^2}{4\tau^2}\varepsilon^{2\alpha}]$, so $\gL(r)-\gL^\ast\ge\min_{t\in[0,\,L_u^2\varepsilon^{2\alpha}/(4\tau^2)]}f(t)$.
Since $f$ is affine, the minimum is attained at an endpoint. At $t=0$ if $\beta\le 2$ and at $t=\frac{L_u^2}{4\tau^2}\varepsilon^{2\alpha}$ if $\beta>2$. In either case,
\[
\gL(r)-\gL^\ast\ge\tfrac{1}{2}\,\ell_\delta^2\,\bar P_\infty^{\inf,2\delta}-(\beta-2)^{+}\,\tfrac{L_u^2}{4\tau^2}\,\varepsilon^{2\alpha}.
\]
Substituting $\beta$ completes the proof.
\end{proof}
\mlpeps*
\begin{proof}
Let $\{y_1,\ldots,y_{N_\varepsilon}\}$ be an $\varepsilon/4$-net in $\gY$, so the open balls $B_i=\{y\in\gY:d_\gY(y,y_i)<\varepsilon/4\}$ cover $\gY$.

For $i=1,\ldots,N_\varepsilon$, define the continuous tent function $u_i(y)=(\varepsilon/4-d_\gY(y,y_i))_+$, which is positive precisely on $B_i$. Let $S(y)=\sum_{j=1}^{N_\varepsilon}u_j(y)$; the cover property gives $S(y)>0$ for every $y\in\gY$, and by compactness $\inf_{y\in\gY}S(y)>0$. Define
\[
\hat\psi:\gY\to\Delta^{N_\varepsilon-1}\subseteq\R^{N_\varepsilon},\qquad\hat\psi(y)_i=u_i(y)/S(y).
\]
Then $\hat\psi$ is continuous on $\gY$, $\hat\psi(y)_i>0$ iff $y\in B_i$, and $\sum_i\hat\psi(y)_i=1$ for every $y$.

Since $\bar\psi:=\hat\psi\circ e^{-1}$ is continuous, we can extend it to a continuous function $\bar\psi:\R^D\to\R^{N_\varepsilon}$.

Apply Theorem~\ref{thm:uat} to $\bar\psi$ on the compact set $e(\gY)$ with tolerance $\eta=\frac{1}{3 N_\varepsilon}$.  There exists a ReLU MLP $\tilde f:\R^D\to\R^{N_\varepsilon}$ such that $\sup_{z\in e(\gY)}\|\tilde f(z)-\bar\psi(z)\|_\infty\le\eta$. If $d>N_\varepsilon$, pad $\tilde f$ with zero coordinates to obtain $f:\R^D\to\R^d$.

Suppose $f(e(y))=f(e(\tilde y))$. The padded coordinates carry no information, so $\tilde f(e(y))=\tilde f(e(\tilde y))$. Hence, 
\begin{align*}
    \|\hat\psi(y)-\hat\psi(\tilde y)\|_\infty
=  \ & \|\bar\psi(e(y))-\bar\psi(e(\tilde y))\|_\infty \\
\leq \ & \|\bar\psi(e(y))-\tilde f(e(y))\|_\infty
+\|\tilde f(e(y))-\tilde f(e(\tilde y))\|_\infty
+\|\tilde f(e(\tilde y))-\bar\psi(e(\tilde y))\|_\infty \,.
\end{align*}
Since $\tilde f(e(y))=\tilde f(e(\tilde y))$, and both $e(y),e(\tilde y)\in e(\gY)$, we have $\|\tilde f(e(y))-\tilde f(e(\tilde y))\|_\infty=0$, $\|\bar\psi(e(y))-\tilde f(e(y))\|_\infty\leq\eta$, $\|\tilde f(e(\tilde y))-\bar\psi(e(\tilde y))\|_\infty\leq\eta$.
Therefore
\[
\|\hat\psi(y)-\hat\psi(\tilde y)\|_\infty\le 2\eta=\frac{2}{3N_\varepsilon}.
\]
We claim that there exists $i^\ast$ with $\hat\psi(y)_{i^\ast}>0$ and $\hat\psi(\tilde y)_{i^\ast}>0$. Suppose for contradiction that $\hat\psi(y)$ and $\hat\psi(\tilde y)$ have disjoint supports. Both are probability vectors in $\R^{N_\varepsilon}$, so $\max_i\hat\psi(y)_i\ge 1/N_\varepsilon$ at some index $i_0$. At that index, $\hat\psi(\tilde y)_{i_0}=0$, hence $|\hat\psi(y)_{i_0}-\hat\psi(\tilde y)_{i_0}|\ge 1/N_\varepsilon>2/(3 N_\varepsilon)$, a contradiction.

Therefore $y\in B_{i^\ast}$ and $\tilde y\in B_{i^\ast}$, so $d_\gY(y,\tilde y)\le d_\gY(y,y_{i^\ast})+d_\gY(y_{i^\ast},\tilde y)<\varepsilon/2\le\varepsilon$.
\end{proof}

\newpage
\section{Excess Loss Upper bound}\label{appendix:upper}
\begin{thm}\label{thm:upper}
Suppose $\tilde p_{\phi_d}(x,y,y^\prime)\in(0,1)$ almost surely and define
\[
  \kappa_d(x,y,y^\prime)=\tilde p_{\phi_d}(x,y,y^\prime)\bigl(1-\tilde p_{\phi_d}(x,y,y^\prime)\bigr).
\]
Define $\nu_{d,x}:=\phi_{d\star}\mu_x$ and $\nu_{d,x}^{\otimes2}:=\nu_{d,x}\otimes\nu_{d,x}$. Define the score function $W_d(x,y,y^\prime)=F^{-1}(\tilde p_{\phi_d}(x,y,y^\prime))$ and the subspace
\[
  \gV_x=\overline{\bigl\{\Delta h \text{ where }h:\gZ_d(x)\to\R\bigr\}}^{L^2(\nu_{d,x}^{\otimes2})},
  \qquad(\Delta h)(z,z')=h(z)-h(z'),
\]
let $\Pi_{\gV_x}$ denote orthogonal projection onto $\gV_x\subseteq L^2(\nu_{d,x}^{\otimes2})$, and set
\[
  U_x:=\widetilde W_{d,x}-\Pi_{\gV_x}\widetilde W_{d,x},
    \qquad
  \widetilde W_{d,x}(z,z')=W_d(x,y,y^\prime)\ \text{where }z=\phi_d(x,y)\,, z'=\phi_d(x,y^\prime).
\]
Then for every $\varepsilon>0$, there exists a measurable prompt-indexed score head $g_x^\varepsilon:\R^d\to\R$ such that $r_d^\varepsilon(x,y)=g_x^\varepsilon(\phi_d(x,y))$ satisfies
\[
  \gL(r_d^\varepsilon)-\gL^\ast \leq
  \E\left[\frac{\bigl(\bar p-\tilde p_{\phi_d}\bigr)^2}{\kappa_d}\right]
  + \frac{1}{8}\,\E_x\bigl[\|U_x\|_{L^2(\nu_{d,x}^{\otimes2})}^2\bigr]
  +\varepsilon\,.
\]
\end{thm}

\begin{proof}
Throughout let $(y,y^\prime)\sim\mu_x^{\otimes2}$, $(z,z')=\bigl(\phi_d(x,y),\phi_d(x,y^\prime)\bigr)\sim\nu_{d,x}^{\otimes2}$, and $\nu_{d,x}=\phi_{d\star}\mu_x$ be the pushforward.

For $u\in[0,1]$ and $v\in(0,1)$, the elementary inequality $\log(1+t)\le t$ yields
$\KL(u\|v)\le\chi^2(u,v)=\frac{(u-v)^2}{v(1-v)}$. With $v=\tilde p_{\phi_d}(x,y,y^\prime)$ and $\kappa_d=v(1-v)$,
\[
  \gE_{\mathrm{emb}}(\phi_d)
  =\E\bigl[\KL(\bar p\|\tilde p_{\phi_d})\bigr]
  \leq
  \E\left[\frac{(\bar p-\tilde p_{\phi_d})^2}{\kappa_d}\right].
\]

Define $W_d(x,y,y^\prime)=F^{-1}(\tilde p_{\phi_d}(x,y,y^\prime))$ and
\[
  \gV_x=\overline{\bigl\{\Delta h: h:\gZ_d(x)\to\R\bigr\}}^{L^2(\nu_{d,x}^{\otimes2})},
  \qquad
  (\Delta h)(z,z')=h(z)-h(z').
\]
Let $\Pi_{\gV_x}$ denote orthogonal projection onto $\gV_x\subseteq L^2(\nu_{d,x}^{\otimes2})$ and set $U_x=\widetilde W_{d,x}-\Pi_{\gV_x}\widetilde W_{d,x}$, where $\widetilde W_{d,x}(z,z')=W_d(x,y,y^\prime)$ is well-defined because $W_d$ is $\sigma(\phi_d(x,y),\phi_d(x,y^\prime))$-measurable.

Fix $\varepsilon>0$. Since $\Pi_{\gV_x}\widetilde W_{d,x}\in\gV_x=\overline{\{\Delta h\}}$, for each $x$ there exists measurable $h_x^\varepsilon:\gZ_d(x)\to\R$ such that
\[
  \|\Pi_{\gV_x}\widetilde W_{d,x}-\Delta h_x^\varepsilon\|_{L^2(\nu_{d,x}^{\otimes2})}^2\le 8\varepsilon.
\]
Extend $h_x^\varepsilon$ arbitrarily to $g_x^\varepsilon:\R^d\to\R$ and set $r_d^\varepsilon(x,y)=g_x^\varepsilon(\phi_d(x,y))$. Then $\Delta r_d^\varepsilon(x,y,y^\prime)=\Delta h_x^\varepsilon\bigl(\phi_d(x,y),\phi_d(x,y^\prime)\bigr)$, and by the Pythagorean identity for orthogonal projection,
\begin{align*}
    \E_{(y,y^\prime)\sim\mu_x^{\otimes2}}\bigl[(W_d-\Delta r_d^\varepsilon)^2\bigr]
= \ & \|\widetilde W_{d,x}-\Delta h_x^\varepsilon\|_{L^2(\nu_{d,x}^{\otimes2})}^2 \\
= \ & \|U_x\|_{L^2(\nu_{d,x}^{\otimes2})}^2+\|\Pi_{\gV_x}\widetilde W_{d,x}-\Delta h_x^\varepsilon\|_{L^2(\nu_{d,x}^{\otimes2})}^2 \\
\leq \ & \|U_x\|_{L^2(\nu_{d,x}^{\otimes2})}^2+8\varepsilon\,.
\end{align*}

Set $v=\tilde p_{\phi_d}(x,y,y^\prime)$, $t_v=F^{-1}(v)=W_d(x,y,y^\prime)$ and $t'=\Delta r_d^\varepsilon(x,y,y^\prime)$. Let $f_v(t)=\KL(v\|F(t))$; then $f_v'(t)=F(t)-v$ vanishes at $t_v$ and $f_v''(t)=F(t)(1-F(t))\le\tfrac14$ for all $t\in\R$. Taylor's theorem with remainder at $t_v$ gives, for some $\xi$ between $t'$ and $t_v$,
\[
f_v(t')=f_v(t_v)+f_v'(t_v)(t'-t_v)+\frac{1}{2} f_v''(\xi)\,(t'-t_v)^2\le\tfrac{1}{8}\,(t'-t_v)^2,
\]
since $f_v(t_v)=0$ and $f_v'(t_v)=0$. Applying this pointwise and integrating,
\[
  \gE_{\mathrm{agr}}(\phi_d,r_d^\varepsilon)
  =\E\bigl[\KL(\tilde p_{\phi_d}\|F(\Delta r_d^\varepsilon))\bigr]
  \le\tfrac18\,\E_x\bigl[\|W_d-\Delta r_d^\varepsilon\|_{L^2(\nu_{d,x}^{\otimes2})}^2\bigr]
  \le\tfrac{1}{8}\,\E_x\left[\|U_x\|_{L^2(\nu_{d,x}^{\otimes2})}^2\right]+\varepsilon.
\]
Adding the embedding and agreement bounds yields the stated inequality.
\end{proof}

\newpage
\section{Extension to Direct Preference Optimization}\label{appendix:dpo}
Fix a finite vocabulary $\gV$ and identify $\gY$ with the set of finite sequences over $\gV$. A policy $\pi$ is a conditional distribution over $\gY$ given $x\in\gX$, with the autoregressive factorization
\[
  \pi(y\mid x)=\prod_{t=1}^{|y|}\pi(y_t\mid y_{<t},x),\qquad y_{<t}=(y_1,\ldots,y_{t-1}).
\]
The reference policy $\pi_{\mathrm{ref}}$ is a fixed autoregressive model and $\lambda>0$ is the KL-regularization coefficient. Define
\[
  \Delta_\pi(x,y,y^\prime) =\lambda\left(\log\frac{\pi(y\mid x)}{\pi_{\mathrm{ref}}(y\mid x)}-\log\frac{\pi(y^\prime\mid x)}{\pi_{\mathrm{ref}}(y^\prime\mid x)}\right).
\]
Then the DPO loss can be written as
\begin{align}\label{eq:dpo-loss}
  \gL_{\mathrm{DPO}}(\pi) =\E_{x\sim\nu}\,\E_{(y,y^\prime)\sim\mu_x^{\otimes 2}}\Bigl[
    -\bar{p}(x,y,y^\prime)\log F\bigl(\Delta_\pi(x,y,y^\prime)\bigr)
    -\bigl(1-\bar{p}(x,y,y^\prime)\bigr)\log\bigl(1-F\bigl(\Delta_\pi(x,y,y^\prime)\bigr)\bigr)
  \Bigr]\,.
\end{align}
Its implicit reward is
\begin{align}\label{eq:r-pi-decomp}
  r_\pi(x,y)
  = \lambda\sum_{t=1}^{|y|}
    \log\frac{\pi(y_t\mid y_{<t},x)}{\pi_{\mathrm{ref}}(y_t\mid y_{<t},x)}.
\end{align}
For each position $t\geq 1$, let $\psi_t:\gX\times\gY\to\R^d$ be a prefix-conditional feature extractor. Define
\[
  \Phi_\psi(x,y)=\bigl(\psi_1(x,y_{\leq 1}),\ldots,\psi_{|y|}(x,y_{\leq|y|})\bigr)\in(\R^d)^{|y|},
\]
and
\[
  \mathcal{R}_{\mathrm{seq}}(\psi):=
  \Bigl\{
    r:\gX\times\gY\to\R
    \,\Big|\,
    r(x,y)=\sum_{t=1}^{|y|}g_t\bigl(\Phi_\psi(x,y_{\leq t})\bigr),
    \ g_t:(\R^d)^t\to\R\ \text{measurable for each }t\ge 1
  \Bigr\}.
\]
For $(y_0,y_0')\sim\mu_x^{\otimes 2}$, set
\[
  \tilde{p}_\psi(x,y,y^\prime) = \E\bigl[\bar{p}(x,y_0,y_0')\big|\,\Phi_\psi(x,y_0)=\Phi_\psi(x,y),\Phi_\psi(x,y_0')=\Phi_\psi(x,y^\prime)\bigr].
\]
Define $P_\psi^\delta$ as $P_d^\delta$ in Definition~\ref{def:condorcet}, with $\tilde{p}_\psi$ in place of $\tilde{p}_{\phi_d}$, and denote the corresponding cycle event by $\mathcal{CC}_\psi^\delta$.

\begin{asmp}[Prefix-covering]\label{asmp:prefix-cover}
For every $x\in\gX$ and any pair of trajectories $(y,y^\prime)$, the set of positions at which they share the same prefix,
\[
  \mathrm{pref}(y,y^\prime)
  =\bigl\{t\in\mathcal{N} : y_{\leq t}=y^\prime_{\leq t}\bigr\},
\]
is determined by the pair of trajectory embeddings $(\Phi_\psi(x,y),\Phi_\psi(x,y^\prime))$. For every fixed $x$ and $y\in\gY$, the map $t\mapsto\Phi_\psi(x,y_{\leq t})$ is injective on $\{0,1,\dots,|y|\}$.
Equivalently, if $\Phi_\psi(x,y_1)=\Phi_\psi(x,y_2)$ then $\mathrm{pref}(y,y_1)=\mathrm{pref}(y,y_2)$ for every trajectory $y$, and two prefixes of the same trajectory at different depths never share an embedding.
\end{asmp}

\begin{restatable}{thm}{dpolb}\label{thm:dpo-lb}
Let $\pi$ satisfy $r_\pi\in\mathcal{R}_{\mathrm{seq}}(\psi)$ and suppose $\tilde p_\psi(x,y,y^\prime)\in(0,1)$ almost surely. Define
\[
  \alpha_\psi(x,y,y^\prime)=\tilde p_\psi(x,y,y^\prime)\bigl(1-\tilde p_\psi(x,y,y^\prime)\bigr),
\]
\[
  \alpha_{\psi}^{\min}(x,y_1,y_2,y_3)=\min\{\alpha_\psi(x,y_1,y_2),\alpha_\psi(x,y_2,y_3),\alpha_\psi(x,y_3,y_1)\},
\]
and
\[
  \bar P_\psi^\delta=\E_{x\sim\nu}\E_{(y_1,y_2,y_3)\sim\mu_x^{\otimes 3}}
  \Bigl[\alpha_\psi^{\min}(x,y_1,y_2,y_3)\,\mathbf{1}_{\mathcal{CC}_\psi^\delta}(x,y_1,y_2,y_3)\Bigr].
\]
Let $e_\psi=\E[(\bar p-\tilde p_\psi)^2]$.
Then for every $\delta\in(0,\tfrac14)$,
\[
  \gL_{\mathrm{DPO}}(\pi)-\gL^\ast
  \geq 8e_\psi+\frac{\ell_\delta^2}{2}\,\bar P_\psi^\delta.
\]
\end{restatable}
\begin{proof}
Set $\mathcal{H}=\sigma\bigl(\Phi_\psi(x,Y),\Phi_\psi(x,Y^\prime)\bigr)$ for $(Y,Y^\prime)\sim\mu_x^{\otimes 2}$ and
\[
  \tilde p_\psi(x,y,y^\prime)=\E\bigl[\bar p(x,Y,Y^\prime)\mid\mathcal{H}\bigr],\qquad
  e_\psi=\E\bigl[(\bar p-\tilde p_\psi)^2\bigr].
\]
With $W_\psi=F^{-1}(\tilde p_\psi)$ and $\alpha_\psi=\tilde p_\psi(1-\tilde p_\psi)$, the same decomposition argument as in Proposition~\ref{prop:decomposition} gives
\[
  \gL_{\mathrm{DPO}}(\pi)-\gL^\ast
  =\gE_{\mathrm{emb}}(\Phi_\psi)+\gE_{\mathrm{agr}}(\Phi_\psi,r_\pi),
\]
\[
  \gE_{\mathrm{emb}}(\Phi_\psi)\ge 8e_\psi.
\]
For the agreement term, pointwise strong convexity of $t\mapsto\KL(p\|F(t))$ at $t_p=F^{-1}(p)$ with curvature $p(1-p)$ yields
\[
  \KL\bigl(\tilde p_\psi\|F(\Delta r_\pi)\bigr)\ge \frac{\alpha_\psi}{2}\,\bigl(W_\psi-\Delta r_\pi\bigr)^2.
\]
Integrating and taking the infimum over $r_\pi\in\mathcal{R}_{\mathrm{seq}}(\psi)$,
\[
  \inf_{r_\pi}\gE_{\mathrm{agr}}(\Phi_\psi,r_\pi)
  \ge \frac12\,\E\left[\alpha_\psi\,\operatorname{dist}\bigl(W_\psi,\gV_{\mathrm{seq}}(x)\bigr)^2\right].
\]
Dropping the additivity restriction can only decrease distance, so
\[
  \operatorname{dist}\bigl(W_\psi,\gV_{\mathrm{seq}}(x)\bigr)\ge \operatorname{dist}\bigl(W_\psi,\overline{\gV_x}\bigr),
\]
and applying the weighted cycle argument from Proposition~\ref{thm:cyclebound} and Lemma~\ref{lem:condorcetlb} with $\tilde p_\psi$ in place of $\tilde p_{\phi_d}$ gives
\[
  \inf_{r_\pi}\gE_{\mathrm{agr}}(\Phi_\psi,r_\pi)\ge\frac{\ell_\delta^2}{2}\,\bar P_\psi^\delta.
\]
Combining with $\gE_{\mathrm{emb}}(\Phi_\psi)\ge 8e_\psi$ yields the claim.
\end{proof}

\begin{restatable}{lem}{rhozero}\label{lem:rho-zero}
If Assumption~\ref{asmp:prefix-cover} holds, then $\gV_{\mathrm{seq}}(x)=\overline{\gV_x}$ for all $x\in\gX$, so the sequential-additivity constraint introduces no additional projection gap.
\end{restatable}
\begin{proof}
It suffices to show that every trajectory reward $r(x,y)=h(\Phi_\psi(x,y))$ can be written as a depth-indexed sum of per-token contributions, i.e., that the resulting $\Delta r$ lies in $\gV_{\mathrm{seq}}(x)$. Combined with the inclusion $\gV_{\mathrm{seq}}(x)\subseteq\overline{\gV_x}$ noted earlier, this gives $\overline{\gV_x}=\gV_{\mathrm{seq}}(x)$.

For each $t\ge 0$ let
\[
  \gS_x^{(t)}=\bigl\{\Phi_\psi(x,y_{\leq t}):y\in\gY,\ |y|\ge t\bigr\}\subseteq(\R^d)^{t},
\]
with $\gS_x^{(0)}=\{\bot\}$ a singleton root. Assumption~\ref{asmp:prefix-cover} states that the map $t\mapsto\Phi_\psi(x,y_{\leq t})$ is injective on $\{0,1,\ldots,|y|\}$ for every $y$ and that prefix-coincidence is determined by the embeddings. In particular, any two trajectories sharing $\Phi_\psi(x,y_{\leq t})=\Phi_\psi(x,y'_{\leq t'})$ must satisfy $t=t'$ and $y_{\leq t}=y'_{\leq t}$. Therefore we work on the disjoint union of depth layers $\bigsqcup_{t\ge 0}\gS_x^{(t)}$, whose elements are depth-tagged states $(t,s)$ with $s\in\gS_x^{(t)}$. The parent map takes one such node as input
\[
\pi:\bigsqcup_{t\ge1}\gS_x^{(t)}\to\bigsqcup_{t\ge0}\gS_x^{(t)},\qquad
\pi(t,s)=\bigl(t-1,\Phi_\psi(x,y_{\le t-1})\bigr),
\]
where $s=\Phi_\psi(x,y_{\le t})$.

Let $r(x,y)=h(\Phi_\psi(x,y))$ be any trajectory reward and write $h_t:\gS_x^{(t)}\to\R$ for its restriction to depth-$t$ states (extended by $h_0(\bot)=0$ at the root). Define $g_t:\gS_x^{(t)}\to\R$ for $t\ge 1$ by
\[
g_t(s)=h_t(s)-h_{t-1}(\pi(t,s)).
\]
Extend each $g_t$ arbitrarily (e.g.\ by zero) to a measurable function on $(\R^d)^t$. Telescoping along any prefix path $(s_0,s_1,\ldots,s_T)$ with $s_t=\Phi_\psi(x,y_{\leq t})$ gives, for every trajectory $y$ with $|y|=T$,
\[
h(\Phi_\psi(x,y))=h_T(s_T)-h_0(\bot)=\sum_{t=1}^{T}g_t\bigl(\Phi_\psi(x,y_{\leq t})\bigr).
\]
Hence $r(x,y)=\sum_{t=1}^{|y|}g_t(\Phi_\psi(x,y_{\leq t}))$, which gives $r\in\mathcal{R}_{\mathrm{seq}}(\psi)$. Consequently $\Delta r\in\gV_{\mathrm{seq}}(x)$, and since $r$ was arbitrary, every trajectory difference lies in $\gV_{\mathrm{seq}}(x)$. Taking $L^2$ closures yields $\overline{\gV_x}\subseteq\gV_{\mathrm{seq}}(x)$, and the reverse inclusion is by definition, so $\overline{\gV_x}=\gV_{\mathrm{seq}}(x)$.
\end{proof}

\newpage
\section{Proofs for Section~\ref{sec:joint}}

\begin{prop}\label{prop:emb-gamma}
Let $\mu'$ be obtained from $\mu$ by the soft-update with reward $r$ and temperature $\beta>0$, and let $\phi'_d$ be an embedding for which $r$ is $\sigma(\phi'_d)$-measurable, i.e.\ $r(x,y)=h(\phi'_d(x,y))$ for some measurable $h$. Set $Z_x=\E_{\mu_x}[e^{\beta r(x,\cdot)}]$, $\rho(x,y)=e^{\beta r(x,y)}/Z_x$, and $\Delta p(x,y,y^\prime)=\bar p(x,y,y^\prime)-\tilde p_{\phi'_d}(x,y,y^\prime)$. Define
\[
\Gamma(\mu,\phi'_d;r,\beta)
\;=\;\E_{\gD_\mu}\Bigl[(\Delta p)^2\bigl(\rho(x,y)\rho(x,y^\prime)-1\bigr)\Bigr]
\;=\;\Cov_{\gD_\mu}\bigl((\Delta p)^2,\,\rho(x,y)\rho(x,y^\prime)\bigr).
\]
Then
\[
e_d(\mu',\phi'_d) \;=\; e_d(\mu,\phi'_d)+\Gamma(\mu,\phi'_d;r,\beta)\,.
\]
\end{prop}

\begin{proof}
Fix $x$, and let $Y,Y^\prime$ be i.i.d.\ with law $\mu_x$ on $\gY$. Set $G_x=\sigma\bigl(\phi'_d(x,Y),\,\phi'_d(x,Y^\prime)\bigr)$. The conditional density of $Y$ under $\mu'_x$ given $\phi'_d(x,Y)=z$ relative to that under $\mu_x$ is
\[
\frac{e^{\beta r(x,y)}}{\E_{\mu_x}\bigl[e^{\beta r(x,Y)}\mid\phi'_d(x,Y)=z\bigr]}=\frac{e^{\beta h(z)}}{e^{\beta h(z)}}=1,
\]
where we used that $r(x,y)=h(\phi'_d(x,y))=h(z)$ is constant on the level set $\{\phi'_d(x,\cdot)=z\}$. Hence the conditional law of $Y$ given $\phi'_d(x,Y)$ is the same under $\mu_x$ and $\mu'_x$; by independence of $(Y,Y^\prime)$ under both product measures, the same holds for the joint conditional given $G_x$. Consequently
\[
\E_{\mu'}[\bar p\mid \gG_x]=\E_\mu[\bar p\mid \gG_x]=\tilde p_{\phi'_d},
\]
so $\Delta p$ is well-defined as the same residual under either measure. Using the change of measure $\frac{d\gD_{\mu'}}{d\gD_\mu}(x,y,y^\prime)=\rho(x,y)\rho(x,y^\prime)$,
\begin{align*}
e_d(\mu',\phi'_d)
&=\E_{\gD_{\mu'}}[(\Delta p)^2]
=\E_{\gD_\mu}\bigl[(\Delta p)^2\,\rho(x,y)\rho(x,y^\prime)\bigr]\\
&=\E_{\gD_\mu}\bigl[(\Delta p)^2\bigr]+\E_{\gD_\mu}\bigl[(\Delta p)^2(\rho(x,y)\rho(x,y^\prime)-1)\bigr]\\
&=e_d(\mu,\phi'_d)+\Gamma(\mu,\phi'_d;r,\beta).
\end{align*}
The covariance form follows from $\E_{\gD_\mu}[\rho(x,y)\rho(x,y^\prime)]=\E_{x\sim\nu}\bigl[(\E_{\mu_x}\rho)^2\bigr]=1$, since $\E_{\mu_x}\rho=Z_x^{-1}\E_{\mu_x}[e^{\beta r}]=1$.
\end{proof}

The change in embedding error decomposes exactly as a covariance under $\gD_\mu$ between the squared embedding residual $(\Delta p)^2$ and the relative likelihood $\rho(x,y)\rho(x,y^\prime)$ that the soft-update assigns to the pair. Embedding error grows under retraining if and only if these two quantities are non-negatively correlated. The next proposition gives sufficient conditions for that to happen.

\begin{prop}\label{prop:gamma}
Under the assumptions of Proposition~\ref{prop:emb-gamma}, $\Gamma(\mu,\phi'_d;r,\beta)\ge 0$ holds in either of the following two settings:
\begin{enumerate}
  \item $\Cov_{\gD_\mu}\bigl((\bar p-\tilde p_{\phi'_d})^2,\,\rho(x,y)\rho(x,y^\prime)\bigr)\ge 0$.
  \item  For $\nu$-almost every $x$, whenever one pair $(y,y^\prime)$ has larger unresolved embedding error $(\bar p-\tilde p_{\phi'_d})^2$ than another pair, it also has no smaller total reward score $r(x,y)+r(x,y^\prime)$, after grouping pairs by the same embedding cells $(\phi'_d(x,y),\phi'_d(x,y^\prime))$.
\end{enumerate}
\end{prop}

\begin{proof}
Setting (1) is the covariance form of $\Gamma$ in Proposition~\ref{prop:emb-gamma}.

For (2), the function $S(x,y,y^\prime)=\rho(x,y)\rho(x,y^\prime)=Z_x^{-2}\,e^{\beta(r(x,y)+r(x,y^\prime))}$ is, conditional on $x$, a non-decreasing function of $r(x,y)+r(x,y^\prime)$ since $\beta>0$. Co-monotonicity with $(\Delta p)^2$ and the Chebyshev sum inequality applied to the pushforward of $\mu_x^{\otimes 2}$ on the common ordering give
\[
\E_{\mu_x^{\otimes 2}}\bigl[(\Delta p)^2\,S\bigr]\ge\E_{\mu_x^{\otimes 2}}[(\Delta p)^2]\cdot\E_{\mu_x^{\otimes 2}}[S]=\E_{\mu_x^{\otimes 2}}[(\Delta p)^2],
\]
using $\E_{\mu_x}[\rho]=1$. Integrating over $x\sim\nu$ yields $\Gamma\ge 0$.
\end{proof}

\jointcov*
\begin{proof}
The $\Gamma$ term is the statement of Proposition~\ref{prop:emb-gamma}. For $\Theta$, fix $x$ and use the derivative
\[
\frac{d(\mu_x^{(t+1)})^{\otimes 3}}{d(\mu_x^{(t)})^{\otimes 3}}(y_1,y_2,y_3)
=\rho_t(x,y_1)\rho_t(x,y_2)\rho_t(x,y_3),
\]
which integrates to $1$ since $\E_{\mu_x^{(t)}}[\rho_t(x,\cdot)]=1$ and the three coordinates are independent. Because $\mathcal{CC}_\infty^{2\delta}$ is defined via $\bar p$ and does not depend on $\mu^{(t)}$,
\begin{align*}
P_\infty^{2\delta,(t+1)}-P_\infty^{2\delta,(t)}
&=\E_\nu\Bigl[\E_{\mu_x^{(t),\otimes 3}}\bigl[(\rho_t(x,y_1)\rho_t(x,y_2)\rho_t(x,y_3)-1)\,\mathbf{1}_{\mathcal{CC}_\infty^{2\delta}}\bigr]\Bigr]\\
&=\E_\nu\Bigl[\Cov_{\mu_x^{(t),\otimes 3}}\bigl(\mathbf{1}_{\mathcal{CC}_\infty^{2\delta}},\rho_t(x,y_1)\rho_t(x,y_2)\rho_t(x,y_3)\bigr)\Bigr],
\end{align*}
where the second equality uses $\E_{\mu_x^{(t),\otimes 3}}[\rho_t(x,y_1)\rho_t(x,y_2)\rho_t(x,y_3)]=1$.
\end{proof}

\jointfixed*
\begin{proof}
Invariance of $\mu^*$ under the soft-update means $d\mu_x^{*,'}/d\mu_x^*=1$, i.e.\ $\rho^*(x,y)=1$ for $\nu$-a.e.\ $x$ and $\mu_x^*$-a.e.\ $y$, in which case both products $\rho^*(x,y)\rho^*(x,y^\prime)$ and $\rho^*(x,y_1)\rho^*(x,y_2)\rho^*(x,y_3)$ are constant (equal to $1$) on the supports of the corresponding product measures. Constants have zero covariance with any function, so $\Gamma(\mu^*,\phi^*;r^*,\beta)=0$ and $\Theta(\mu^*;r^*,\beta)=0$.
\end{proof}

\newpage
\section{More Experiment Details and Results}\label{appendix:exp}

\paragraph{Fitting the Reward}
For each $d$, items are partitioned into groups by $\phi_d$.
We fit one scalar score per group by full-batch gradient descent on the pairwise logistic loss to the coarsened matrix $\tilde P_{\phi_d}$, with learning rate $5\times 10^{-2}$ and a decaying step-size schedule.

\paragraph{DPO}
We define implicit rewards $r_i=\lambda\bigl(\log\pi_\theta(y_i)-\log\pi_{\mathrm{ref}}(y_i)\bigr)$, average $r$ within each $\phi_d$-group, and set $q_{ij}^\pi=\sigma(r_i-r_j)$.
The policy $\pi_\theta$ and reference $\pi_{\mathrm{ref}}$ are fixed in this evaluation. 

\begin{figure}[tbh]
    \centering
    \includegraphics[width=0.48\linewidth]{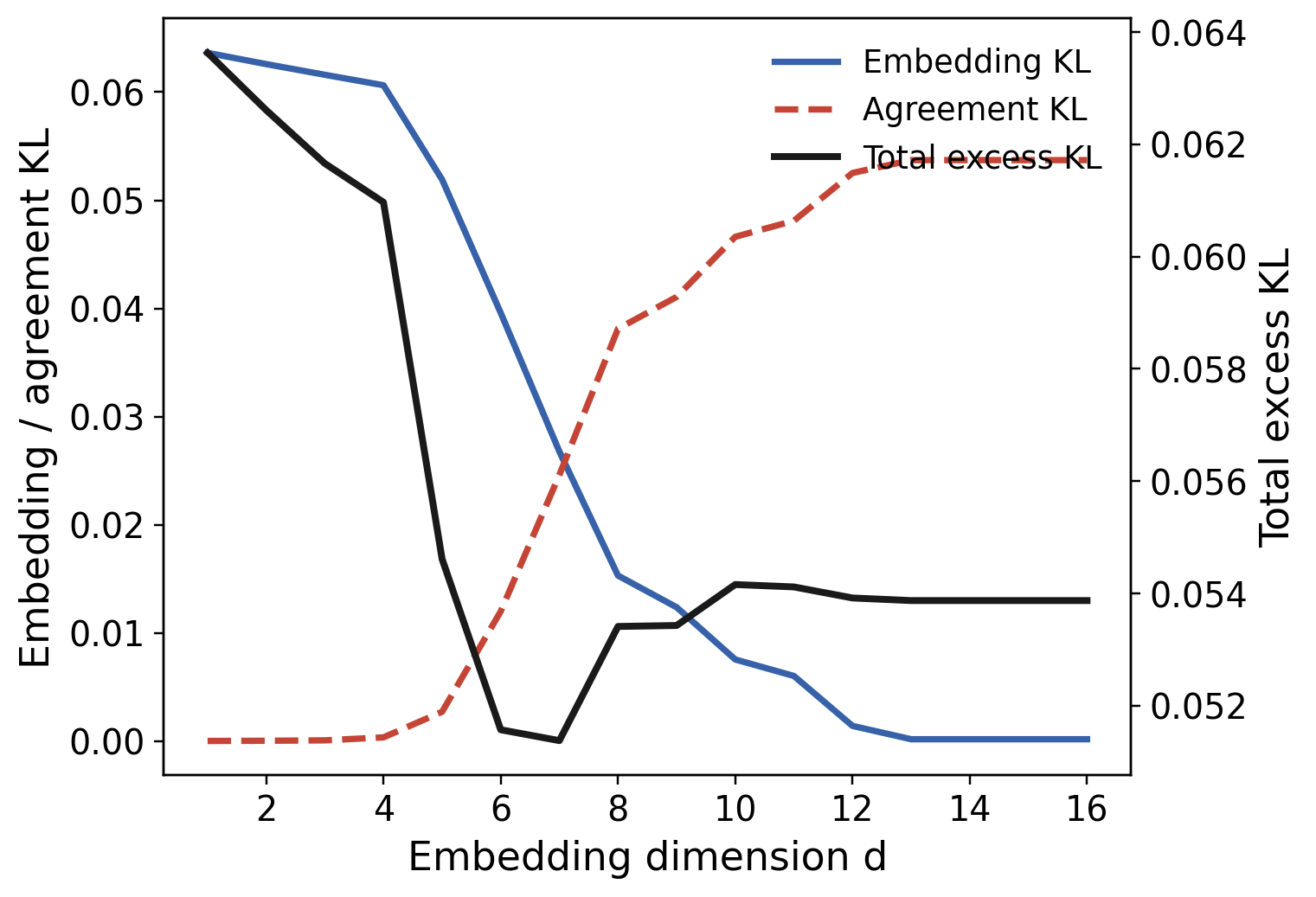}
    \includegraphics[width=0.48\linewidth]{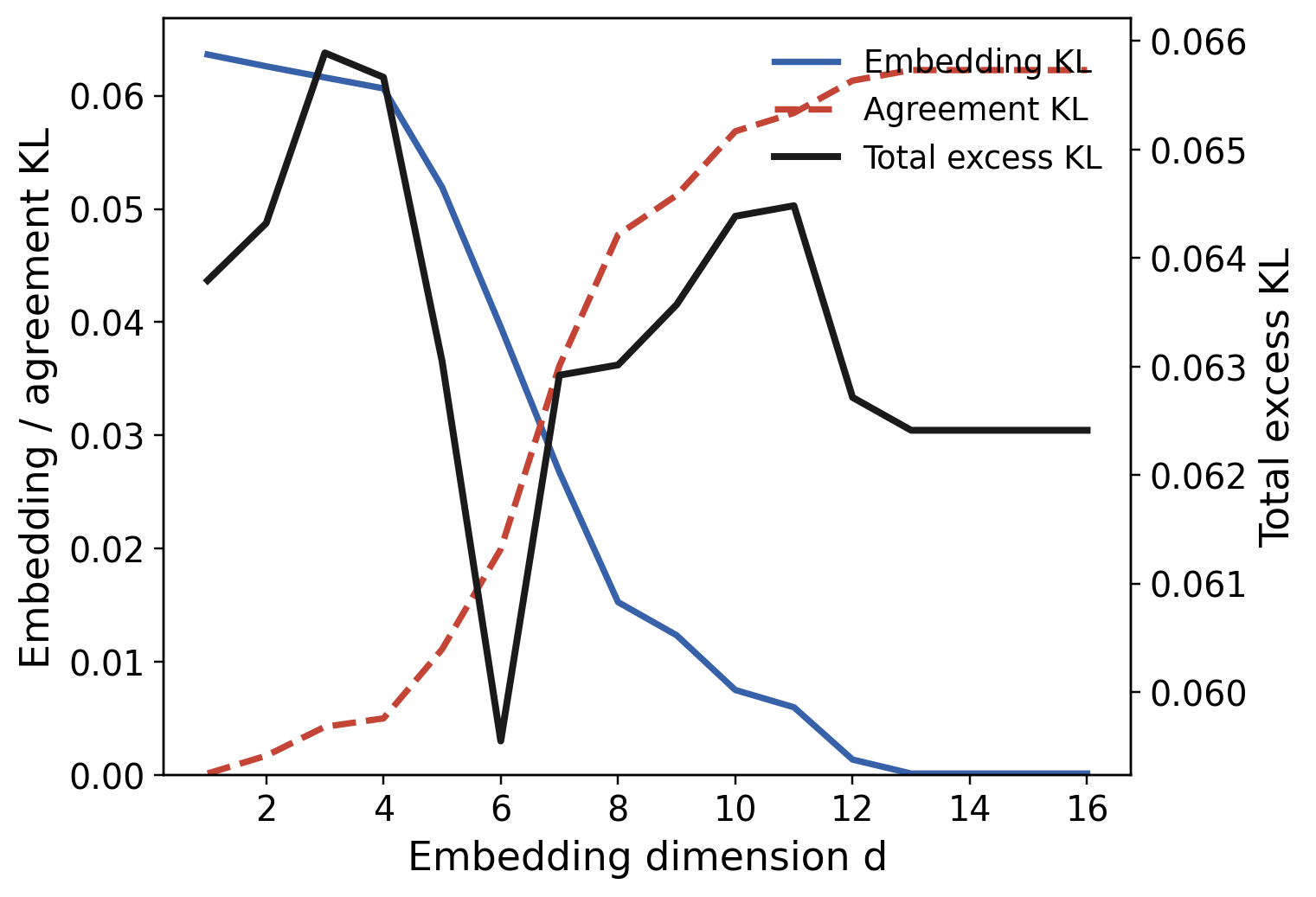}
    \includegraphics[width=0.48\linewidth]{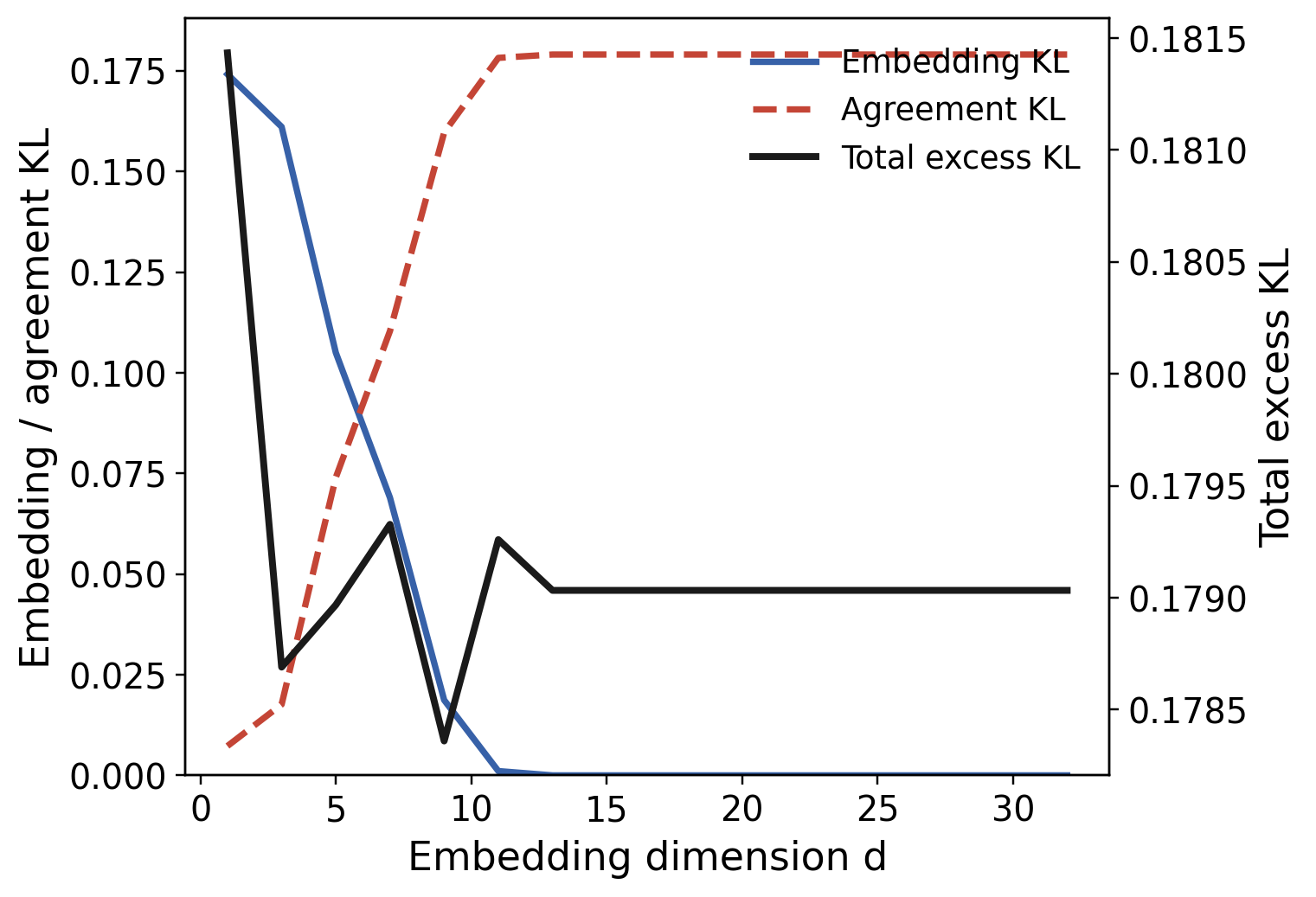}
    \includegraphics[width=0.48\linewidth]{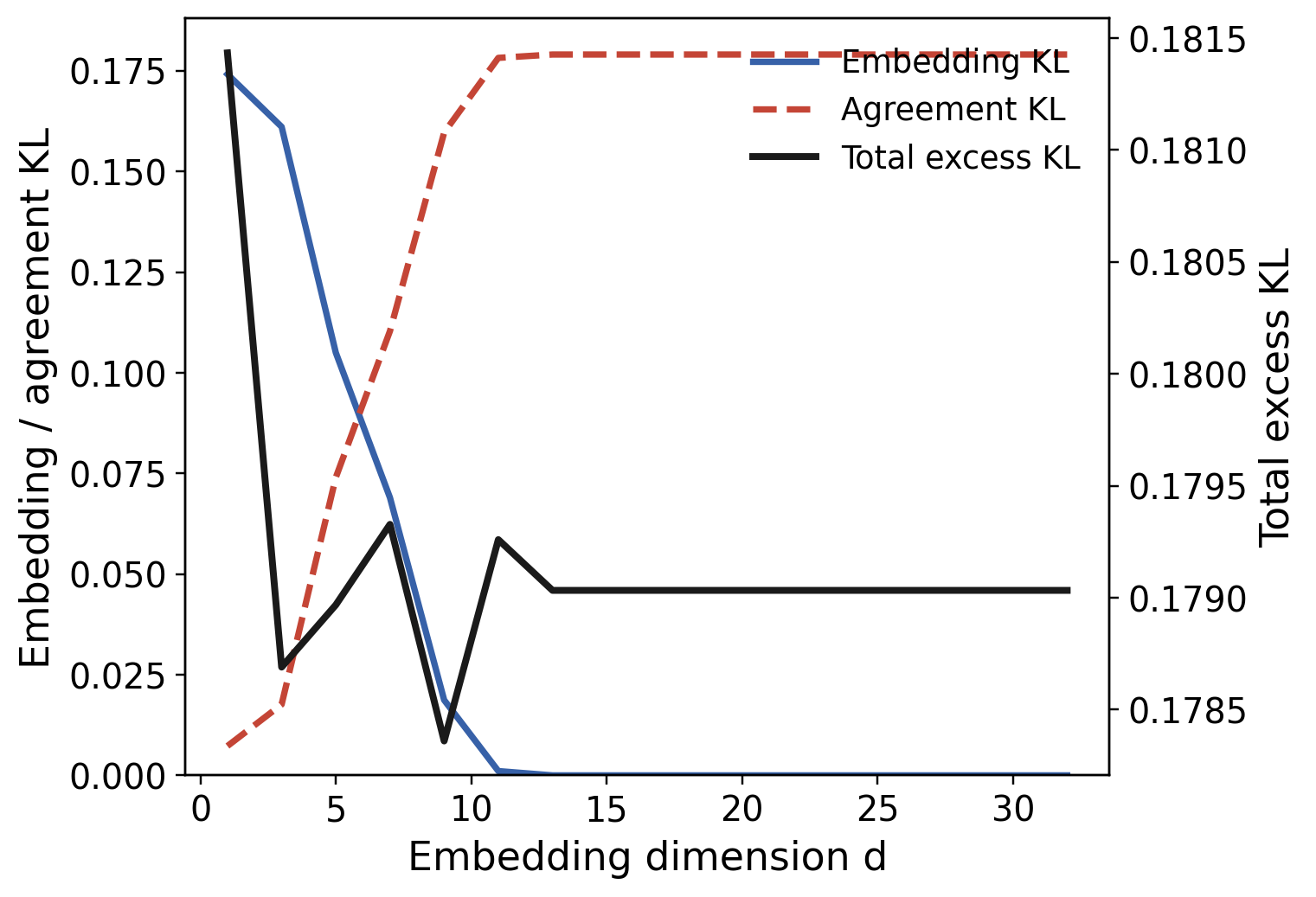}
    \caption{Top row: results on Jester, Bottom row: results on Sushi. Left panel: results with RLHF, right panel: results with DPO.}
    \label{fig:mpnet_appendix}
\end{figure}
\begin{figure}[tbh]
    \centering
    \includegraphics[width=0.48\linewidth]{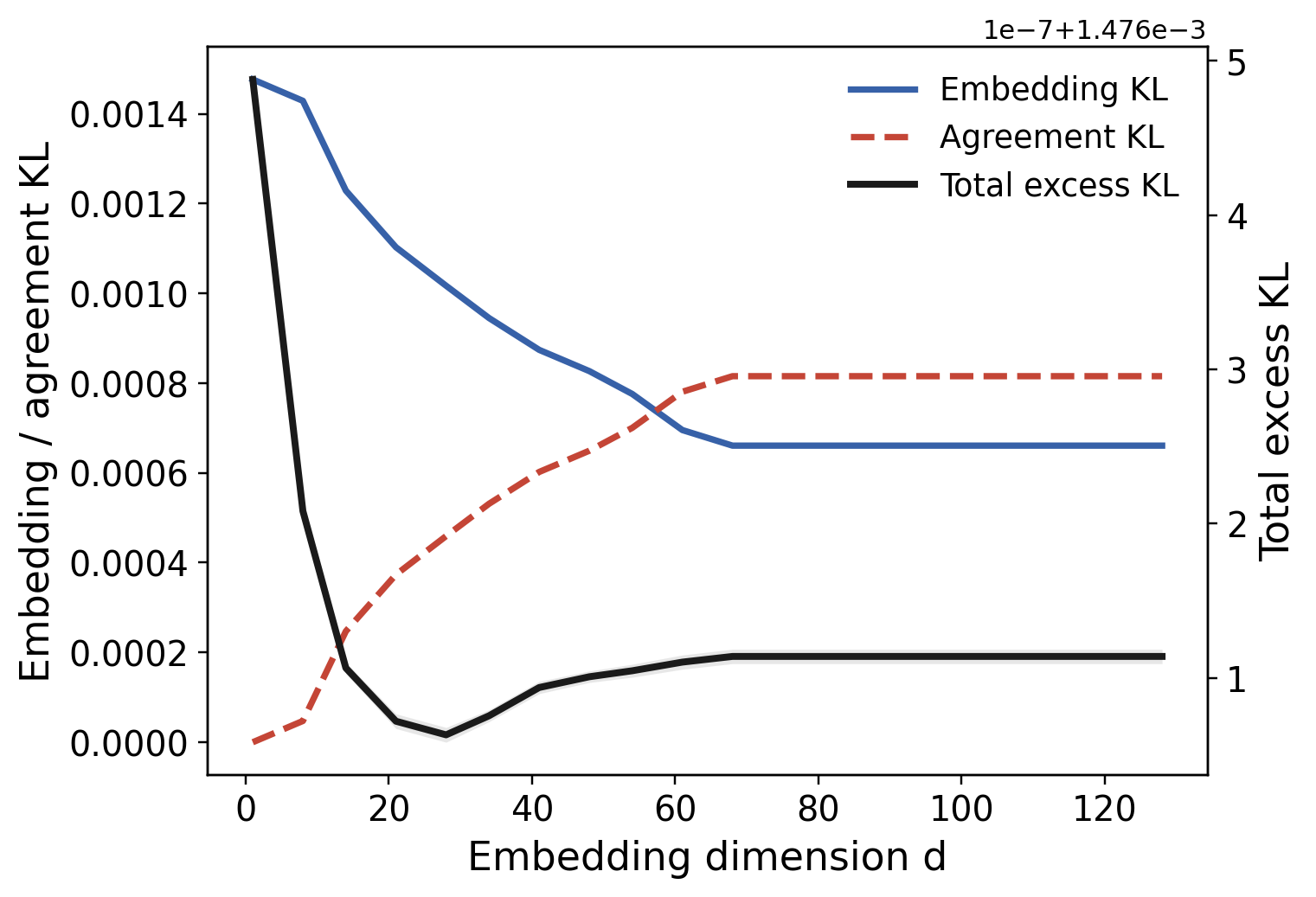}
    \includegraphics[width=0.48\linewidth]{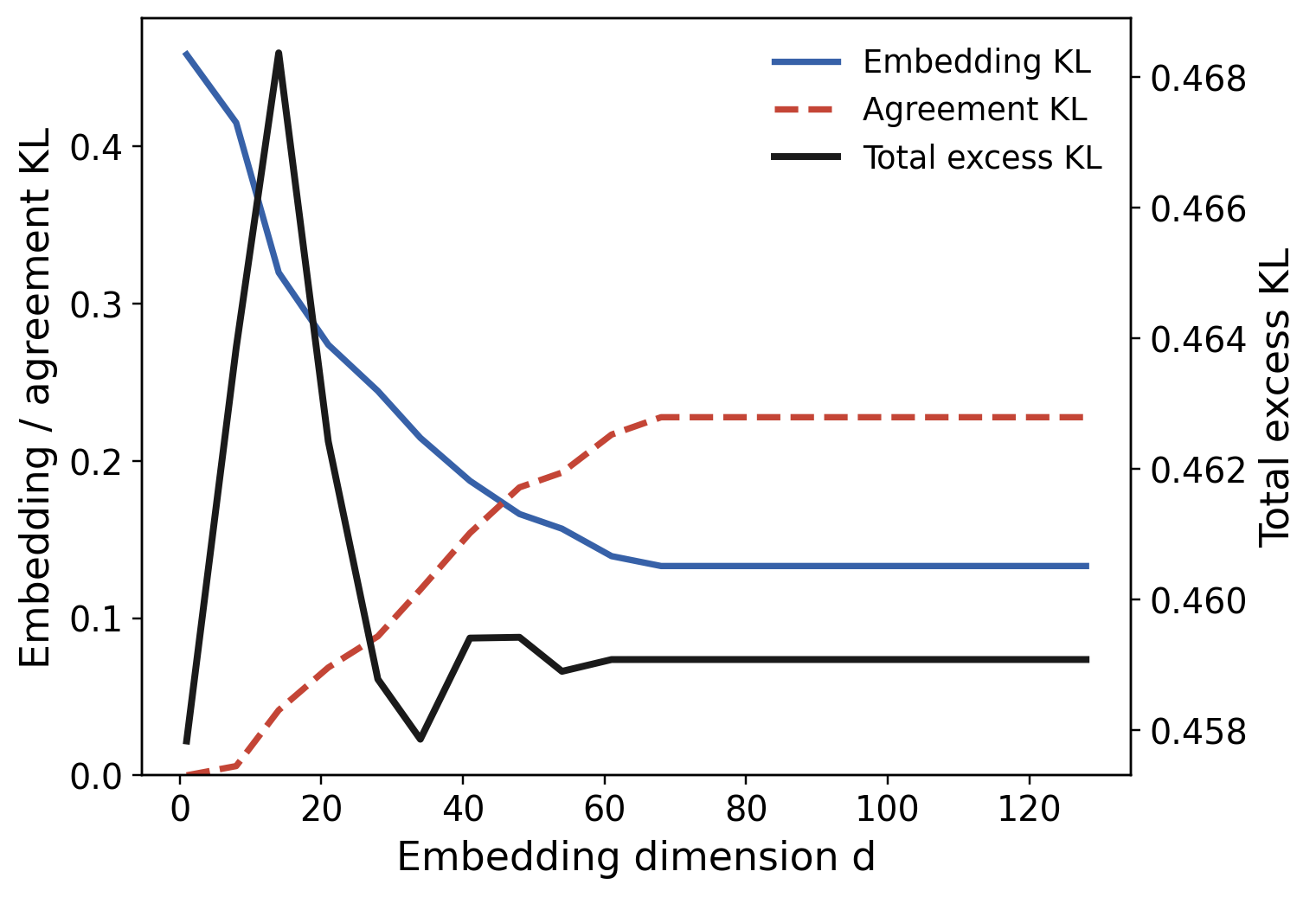}
    \caption{Excess loss decomposition on the MT-Bench preference dataset with \texttt{all-MiniLM-L6-v2} embedding. Left: RLHF. Right: DPO.}
    \label{fig:mtbench}
\end{figure}

\begin{figure}[thb]
    \centering
    \includegraphics[width=\linewidth]{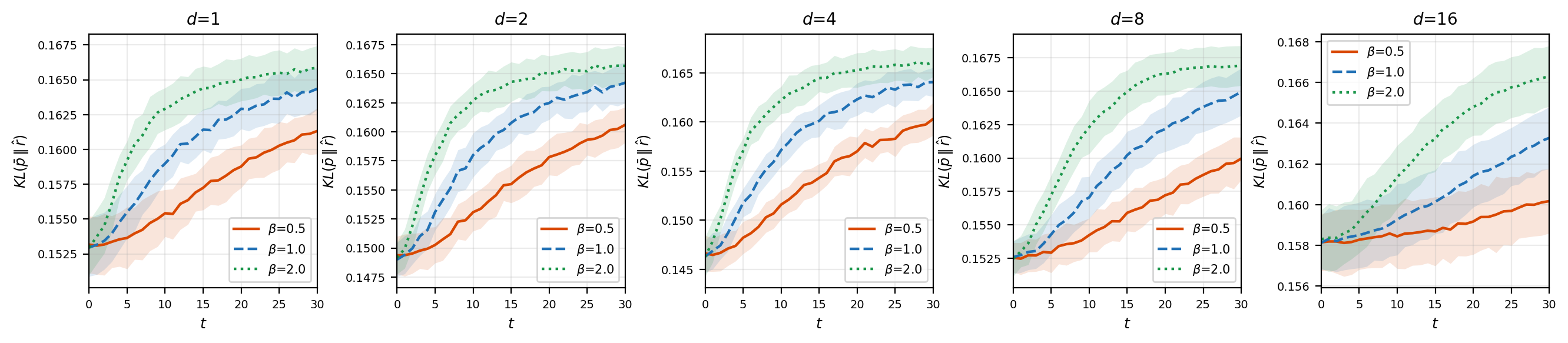}
    \includegraphics[width=\linewidth]{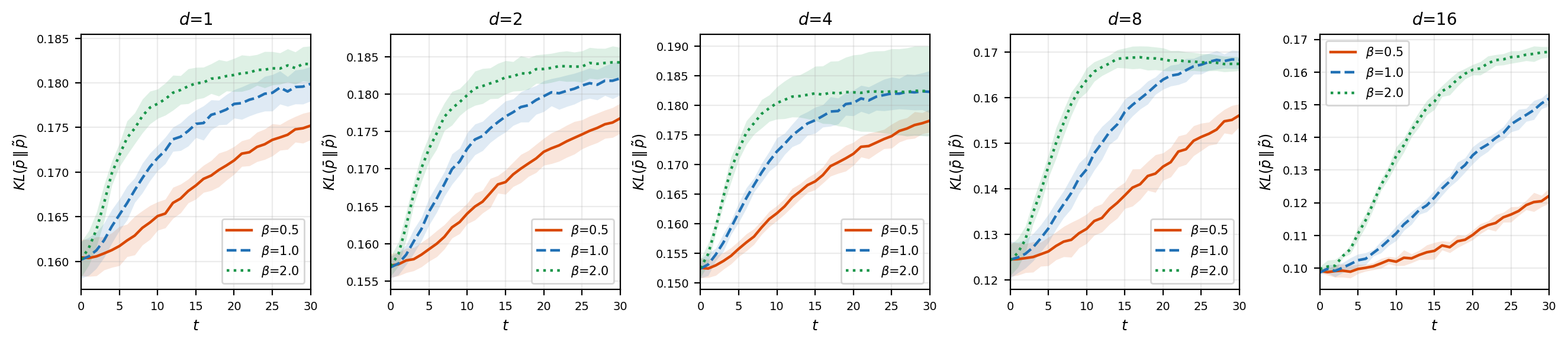}
    \includegraphics[width=\linewidth]{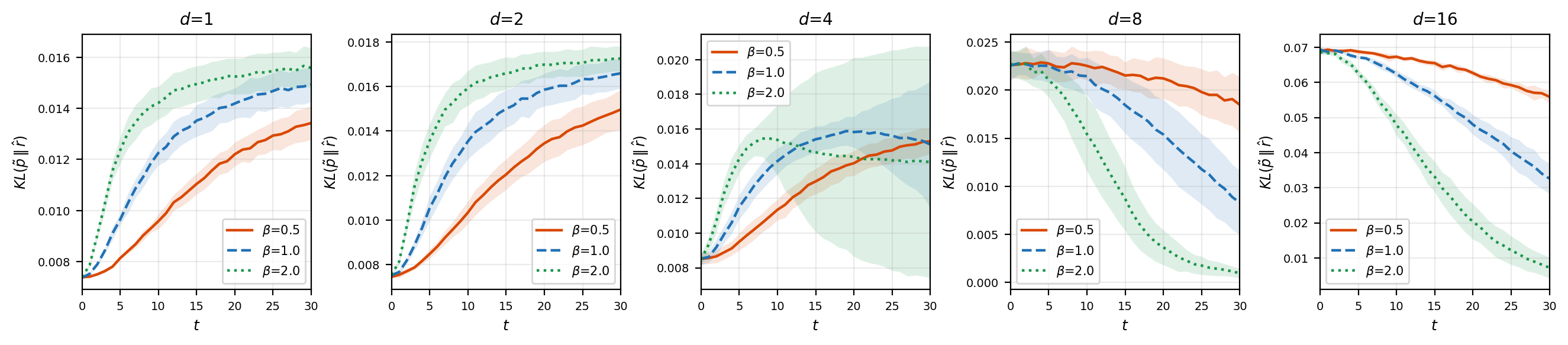}
    \caption{Joint embedding--reward updates on synthetic data. Top: total excess loss over update rounds. Middle: embedding error $\gE_{\mathrm{emb}}(\phi_d)$ over rounds. Bottom: agreement cost $\gE_{\mathrm{agr}}(\phi_d,r)$ over rounds.}
    \label{fig:online}
\end{figure}
\paragraph{Licenses and terms of use.}
Pretrained embeddings are \texttt{sentence-transformers/all-MiniLM-L6-v2} and \texttt{sentence-transformers/all-mpnet-base-v2} from the Hugging Face Hub (\url{https://huggingface.co/sentence-transformers/all-MiniLM-L6-v2}, \url{https://huggingface.co/sentence-transformers/all-mpnet-base-v2}). Each model card lists Apache License, version 2.0, and we use the checkpoints accordingly. MT-Bench pairwise judgments are taken from the Hugging Face dataset \texttt{lmsys/mt\_bench\_human\_judgments} (\url{https://huggingface.co/datasets/lmsys/mt_bench_human_judgments}), which lists CC BY 4.0 on its dataset card. The Jester and Sushi data are used under the terms of the original releases cited above.

Figure~\ref{fig:mpnet_appendix} reports results on Jester and Sushi for both RLHF and DPO with pretrained embedding from \texttt{mpnet} \citep{song2020mpnet}.

Figure~\ref{fig:mtbench} shows the results on MT-Bench. Again, the embedding loss decreases and the agreement cost increases with embedding dimension, and their sum may be minimized at an intermediate dimension.

We also evaluate joint embedding-reward updates using the same synthetic setup. We initialize from a constant embedding and run $T=10$ rounds of alternating reward fitting and distribution tilting. At each round we measure the embedding error $\gE_{\mathrm{emb}}(\phi_d^{(t)})$, the agreement cost $\gE_{\mathrm{agr}}(\phi_d^{(t)},r_t)$, and their sum. The Top panel of Figure~\ref{fig:online} shows that total excess loss increases over round. The middle panel shows that the embedding error grows over rounds. The bottom panel shows agreement cost can either increase or decrease, with different dimension. 

\end{document}